\documentclass[10pt,reqno]{amsart}
\usepackage{amsmath,latexsym,amssymb,ifpdf}
\usepackage{color,graphicx}
\graphicspath{{./}}

     \makeatletter
     \def\section{\@startsection{section}{1}%
     \z@{.7\linespacing\@plus\linespacing}{.5\linespacing}%
     {\bfseries
     \centering
     }}
     \def\@secnumfont{\bfseries}
     \makeatother
\newtheorem{theorem}{Theorem}[section]
\newtheorem{lemma}[theorem]{Lemma}
\newtheorem{proposition}[theorem]{Proposition}
\newtheorem{corollary}[theorem]{Corollary}
\theoremstyle{definition}

\theoremstyle{remark}
\newtheorem{remark}[theorem]{Remark}
\numberwithin{equation}{section} 

\renewcommand{\paragraph}[1]{{\bf #1.}}
\definecolor{myred}{rgb}{0.8,0,0}  

{\vskip\baselineskip\noindent\textbf{Proof of {#1}:}}%
{\hspace*{.1pt}\hspace*{\fill}\BOX\vskip\baselineskip}

\def \R{\mathbb{R}}               
\def \N{\mathbb{N}}               
\def \1{{\bf 1}}                
\def \0{{\bf 0}}

\def\eps{\varepsilon}

\def\var{\text{var}}
\def\cov{\text{cov}}
\def\varmu{\nu}
\def\expert{Z}
\def\varexp{\Gamma}
\def\varexpbound{\overline\Gamma}

\def \1{{\bf 1}}                
\def \0{{\bf 0}}

\begin{document}

\title[Expert Opinions and Logarithmic Utility Maximization]{Expert Opinions and Logarithmic Utility Maximization in a Market with Gaussian Drift}

\author[A.~Gabih]{Abdelali Gabih}
\address{Abdelali Gabih, Universit\'{e} Cadi Ayyad, ENSA Marrakech, Laboratoire OSCARS, Boulevard
Abdelkarim Khattabi Gu\'{e}liz BP 575, 40000 Marrakech, Morocco}
\email{a.gabih@uca.ma}

\author[H.~Kondakji]{Hakam Kondakji}
\address{Hakam Kondakji, Mathematical Institute, Brandenburg
University of Technology Cottbus -- Senftenberg ,  Postfach 101344, D-03013 Cottbus,Germany}
\email{Hakam.Kondakji@tu-cottbus.de}

\author[J.~Sass]{J\"orn Sass}
\address{J\"orn Sass, Department of Mathematics, University of Kaiserslautern,
P.O.Box 3049, 67653 Kaiserslautern, Germany, Germany}
\email{sass@mathematik.uni-kl.de}

\author[R.~Wunderlich]{Ralf Wunderlich}
\address{Ralf Wunderlich, Mathematical Institute, Brandenburg
University of Technology Cottbus -- Senftenberg , Postfach 101344, D-03013 Cottbus,
Germany} \email{ralf.wunderlich@tu-cottbus.de }

\date{ \today}



\begin{abstract}
This paper investigates optimal portfolio strategies in a financial market
where the drift of the stock returns is driven by an unobserved Gaussian mean reverting process. Information
on this process is obtained from observing stock returns and
expert opinions. The latter provide at discrete time
points an unbiased estimate of the current state of the drift. Nevertheless, the drift can only be observed partially and the best estimate is given by the conditional expectation given the available information, i.e., by the filter. We provide the filter equations in the model with expert opinion and derive in detail properties of the conditional variance. For an investor who maximizes expected logarithmic utility of his portfolio, we derive the optimal strategy explicitly in different settings for the available information. The optimal expected utility, the value function of the control problem, depends on the conditional variance. The bounds and asymptotic results for the conditional variances are used to derive bounds and asymptotic properties for the value functions. The results are illustrated with numerical examples.
\end{abstract}

\maketitle


\section{Introduction}

We consider an investor who wants to maximize expected logarithmic utility of terminal wealth obtained by trading in a financial market consisting of one riskless asset and one stock. Stock returns satisfy
\[
R_t = \int_0^t \mu_s\,ds + \sigma\, dW_s,
\]
where $W$ is a Brownian motion, the volatility $\sigma>0$ is constant, but the drift $\mu$ is some stochastic process independent of $W$. Thus the drift is hidden and has to be estimated from the observed stock returns.
The best estimate in a mean-square sense is the filter. While under suitable integrability assumptions we can get quite far in solving the utility maximization problem, see Bj\"ork, Davis and Land\'en \cite{Bjoerk et al (2010)} and  Lakner \cite{Lakner (1995)}, we need models which allow for finite dimensional filters to solve the problem completely including the computation of an optimal policy. Therefore, in the literature the drift process is either modeled as Ornstein-Uhlenbeck process (OUP) or as a  continuous time Markov chain (CTMC). In both models  finite-dimensional filters are well known, the Kalman and Wonham filters, respectively, see e.g.\ Elliott, Aggoun and Moore \cite{Elliott et al. (1994)}, Liptser and Shiryaev \cite{Liptser-Shiryaev}. In these two models the utility maximization problem is solved, see Brendle \cite{Brendle (2006)}, Lakner \cite{Lakner (1998)}, Putsch\"ogl and Sass \cite{Putschoegl and Sass (2008)} and Honda \cite{Honda (2003)}, Rieder and B\"auerle \cite{Rieder_Baeuerle2005}, Sass and Haussmann \cite{Sass and Haussmann (2004)}, respectively.

However, to improve the estimate, an investor may rely on expert opinions. These provide a noisy estimate of the current state of the drift. For unbiased estimates, this reduces the variance of the filter. The better estimate then improves expected utility. This can be seen as a continuous time version of the static Black-Litterman approach which combines an estimate of the asset return vector with expert opinions on the performance of the assets, see Black and Litterman \cite{Black_Litterman (1992)}. For a comparison with other Bayesian and robust Bayesian methods see Sch\"ottle, Werner and Zagst \cite{Schoettle et al. (2010)}.

Frey, Gabih and Wunderlich \cite{Frey et al. (2012), Frey-Gabih-Wunderlich-2014} solve the case of an underlying CTMC.  As an approximation, also expert opinions arriving continuously in time can be introduced. This allows for more explicit solutions for the portfolio optimization problem. Davis and LLeo \cite{Davis and Lleo (2013)} consider this approach for an underlying OUP, Sass, Seifried and Wunderlich \cite{Sass et al (2014)} address the CTMC.

In this paper we look at the remaining case, an underlying OUP with time-discrete expert opinions. Due to the combination of continuous time-observations (stock returns) and discrete-time expert opinions, optimal portfolio policies are quite involved. We expect that for power utility they can be derived along the lines of \cite{Frey et al. (2012), Frey-Gabih-Wunderlich-2014} using a stochastic control approach with an additional policy-dependent change of measure, cf.\ Nagai and Peng \cite{Nagai and Peng (2002)}, and working with viscosity solutions. However, since our focus lies on explicit results and bounds on the improvement by expert opinions for different information regimes, we shall consider only logarithmic utility here. Explicit results for other utility functions are up to future research. On the other hand, an extension for logarithmic utility to the multivariate case, i.e., to markets with more than one risky asset, is straightforward. Filtering results and optimal policies can be derived analogously. But closed form solutions are no longer available for the conditional variances which then have to be computed numerically. Convergence results as in Section \ref{cond_variance} would be more difficult to obtain.

The paper is organized as follows. In Section \ref{market_model} we define the model for an OUP drift process and specify our concept of  expert opinions. We introduce different settings for the available information which arises from observing the stock returns only (classical partial information), from expert opinions only and from the combination of stock returns and expert opinions. As reference we also consider full information. In Section \ref{partial_info} we state the classical Kalman filter for pure return observations and derive in the cases with expert opinion the filtering equations. In Section \ref{cond_variance} we analyze the conditional variance in detail: In addition to staightforward bounds and monotonicity assumptions, Proposition \ref{prop_cond_var_asymN} provides the limits for an increasing number of i.i.d.\ expert opinions for a finite time horizon and Proposition \ref{prop_cond_var_asym_t}  provides tight asymptotic bounds for the conditional variance for regularly arriving expert opinions for an infinte time horizon. These properties and bounds are important since the optimal value is a function of the conditional variance. Our main result is Theorem \ref{opt_value_theorem} which provides for logarithmic utility explicit solutions in all four information settings. In the remainder of Section \ref{optimization} we compare  the optimal expected utilities (value functions) for the different cases. In Section \ref{numerics} we provide extensive  simulations and numerical computations to illustrate our theoretical results.

Summarizing, our contributions lie in (i) finding filtering equations in the settings with expert opinion, (ii) solving the $\log$-utility maximization problem with closed form solutions for optimal policies and values and (iii) deriving limits and bounds for the conditional variance and using these to compare different information settings.

\section{Financial Market Model}

\label{market_model}  For a  fixed date $T>0$ representing the
investment horizon, we work on a filtered probability space
$(\Omega,\mathcal{G},\mathbb{G},P)$, with filtration
$\mathbb{G}=(\mathcal {G}_t)_{t \in [0,T]}$ satisfying the usual
conditions. All processes are assumed to be $\mathbb{G}$-adapted.

\paragraph{Price dynamics}
 We consider a market model  for one risk-free bond with prices
$S^0_t=1$ and  one risky security with prices $S_t$ given by
\begin{eqnarray}
\label{stockmodel} dS_t=S_t\,\Big(\mu_t dt+\sigma dW_t\Big).
\end{eqnarray}
The volatility $\sigma$ is assumed to be a positive constant and $W$
is an one-dimensional $\mathbb{G}$-adapted Brownian motion.  The dynamics of the
drift process $\mu$ are given by the stochastic differential equation (SDE)
\begin{eqnarray}
\label{drift} d\mu_t=\alpha(\delta-\mu_t) dt+\beta dB_t,
\end{eqnarray}
where
$\alpha, \beta >0$ and $\delta \in \R $ are constants  and $B$ is a Brownian motion independent of $W$. Here, $\delta$ is the mean-reversion level, $\alpha$ the mean-reversion speed and $\beta$ describes the volatility of $\mu$.
The initial value $\mu_0$ is assumed to be a normally distributed
random variable independent of $B$ and $W$ with mean $m_0\in \R$ and variance $\varmu_0\ge 0$.
It is well-known that  SDE \eqref{drift} has the closed-form solution
\begin{eqnarray}
\label{mu_explicit}
\mu_t = \delta +e^{-\alpha t}\Big[(\mu_0 -\delta) + \beta  \int_0^t e^{\alpha s} dB_s\Big],\quad t>0.
\end{eqnarray}
This is a  Gaussian process and  known  as  Ornstein-Uhlenbeck
process. It has moments
\begin{align}
\label{mu_mean}
 m_t := E[\mu_t] &=   ~~~~\delta ~~~~~+~~e^{-\alpha t}(m_0 -\delta) &&\text{(mean)}\\
\label{mu_var}
 \varmu_t := \var[\mu_t] &=  \frac{\beta^2}{2\alpha} + e^{-2\alpha t}\Big(\nu_0- \frac{\beta^2}{2\alpha}\Big)
&& \text{(variance)}\\
\nonumber
\cov[\mu_s,\mu_{t}]&=\frac{\beta^2}{2\alpha} e^{-\alpha |t-s|} + e^{-\alpha (t+s)}\Big(\nu_0- \frac{\beta^2}{2\alpha}\Big) && \text{(covariance function)}
\end{align}
for $s,t\ge 0$.
It can be seen, that mean and variance approach exponentially fast  the limits $\delta$ and $\frac{\beta^2}{2\alpha}$, respectively, i.e.~asymptotically for $t\to \infty$ the drift $\mu_t$ has a  $\mathcal{N}(\delta,\frac{\beta^2}{2\alpha})$ distribution which is the stationary distribution.
Starting with the stationary distribution leads to a (strict-sense) stationary drift process $\mu$ with mean $\delta$ and correlation function $\cov[\mu_t,\mu_{t+\tau}]=\frac{\beta^2}{2\alpha} e^{-\alpha |\tau|} $ for $t,t+\tau\ge 0$.

We define the return process $R$  associated with the price process $S$
by $dR_t=dS_t/S_t$.   Note that $R$ satisfies $dR_t = \mu_t dt +
\sigma dW_t $ and  $R_t = \log S_t +\frac{\sigma^2}{2} t$. So we have the equality
$\mathbb{G}^R = \mathbb{G}^{\log S} = \mathbb{G}^S\,.$ This is
useful, since it allows to work with $R$ instead of $S$ in the
filtering part.

\paragraph{Investor information and expert opinions}
An investor cannot observe the drift process $\mu$ directly.
He  has  noisy observations of
the hidden process  $\mu$ at his disposal. More precisely we assume
that the investor  observes the return process $R$ and that he
receives at $N$ discrete  deterministic points in time $t_0,\ldots,t_{N-1}$ with $0=t_0<\ldots<t_{N-1} <T$ and  $N\in \N$ noisy
signals about the current state of $\mu$. These signals or ''views'' are
interpreted as expert opinions and modelled by Gaussian random variables of the form $\expert_k= \mu_{t_k} + \sqrt{\varexp_k} \eps_k$ with i.i.d.~random variables  $\eps_0,\ldots,\eps_{N-1}  \sim\mathcal{N}(0,1)$ independent of the Brownian motions $B$ and $W$. So we assume that the expert's views are unbiased, i.e., in expectation they coincide with the current (and unknown) value of the drift. The variance $\varexp_k$ is a measure for the reliability of the expert: the larger $\varexp_k$ the less reliable is the expert. Note that we always assume that an investor knows the model parameters, in particular the
distribution $\mathcal{N}(m_0,\nu_0)$ of  the initial value $\mu_0$.
Setting  $\varmu_0=0$ we can model   a known  (deterministic)  initial value $m_0$ of  the drift.

The information available to an investor can be described by the \textit{investor filtration} $\mathbb{F}^H=(\mathcal{F}^H_t)_{t\in[0,T]}$ for which we consider four cases $H\in\{R,E,C,F\}$, where
\[\begin{array}{rcll}
\mathbb{F}^R&=& (\mathcal {F}_t^R)_{t \in [0,T]} & \text{with }\mathcal {F}_t^R~\text{generated by }~
\{R_s, s\le t\}, \\[0.5ex]
\mathbb{F}^E&=& (\mathcal {F}_t^E)_{t \in [0,T]} & \text{with }\mathcal {F}_t^E~\text{generated by }~
\{Z_k, t_k\le t\}, \\[0.5ex]
\mathbb{F}^C&=& (\mathcal {F}_t^C)_{t \in [0,T]} & \text{with }\mathcal {F}_t^C~\text{generated by }~
\{R_s, s\le t,\,Z_k, t_k\le t\}, \\[0.5ex]
\mathbb{F}^F&=& \mathbb{G},
\end{array}
\]
and where we assume that the $\sigma$-algebras $\mathcal{F}_t^H$, $H\in\{R,E,C\}$ are augmented by the null sets $\mathcal{N}$ of $P$, e.g., $\mathcal{F}_t^R = \sigma(\{R_s,\, 0 \le s\le t\} \cup \mathcal{N})$.
Note that $\mathcal {F}_t^C=\mathcal {F}_t^R \vee \mathcal {F}_t^E$.
$\mathbb{F}^R$ and $\mathbb{F}^E$ correspond to an investor who observes only returns or expert opinions, respectively. $\mathbb{F}^C$ describes the information arising from the combination of returns and expert opinions. Finally, $\mathbb{F}^F$ describes an investor who has full information on the drift process $\mu$.
For stochastic drift full information is unrealistic, but we use results obtained for $\mathbb{F}^F$ as  reference points for the corresponding results in the other cases, e.g. when defining the efficiency in Section \ref{numerics}.

\section{Partial Information and Filtering}
\label{partial_info}

The filter for the drift  $\mu_t$ is the projection on the $\mathcal{F}_t^R$-measurable random variables. It is given by the conditional expectation $\widehat{\mu}_t=E[\mu_t|\mathcal{F}^H_t]$ and is optimal estimate in the mean-square sense. In the following we discuss the four cases $\mathbb{F}^H$,  $H\in\{R, E, C, F\}$.
\\[1ex]

\paragraph{Return observations only ($H=R$)}
If the investor only observes the returns and has no access to the additional expert opinions,  his information is given by $\mathbb{F}^R$. Then the drift process $\mu$ and return process $R$ are jointly Gaussian and hence the conditional distribution of $\mu$ given $R$  is
completely described by the conditional mean
$\widehat{\mu}_t^R:=E[\mu_t|\mathcal{F}^R_t]$ and the conditional
variance $\gamma_t^R:=E[(\mu_t-\widehat{\mu}^R_t)^2|\mathcal{F}^R_t]$. The dynamics of $\widehat{\mu}^R$ and $\gamma^R$ are given by the well-known Kalman filter, see e.g.~Liptser and Shiryaev  \cite{Liptser-Shiryaev},
which consists of the following SDE for $\widehat{\mu}_t^R$
\begin{eqnarray}
\label{filter_R}
d\widehat{\mu}^R_t = \big(\alpha(\delta-\widehat{\mu}^R_t)- \sigma^{-2}\gamma_t^R\;\widehat{\mu}^R_t\big) dt+\sigma^{-2}\gamma_t^R\;dR_t, \quad \widehat{\mu}^R_0=m_0,
\end{eqnarray}
and a deterministic ODE for the conditional variance $\gamma_t^R$
\begin{eqnarray}
\label{cond_var_R}
\frac{d}{dt}\gamma^R_t&=&-\sigma^{-2}(\gamma^R_t)^{{^2}}-2\alpha\gamma^R_t+\beta^2,
\quad \gamma^R_0=\varmu_0,
\end{eqnarray}
hence $\gamma_t^R$   is deterministic. The above ODE is known as Ricatti Equation and has for initial value $\gamma^R_0=\varmu_0\ge 0$ the unique non-negative solution
(see e.g.~Lakner \cite{Lakner (1998)})
\begin{eqnarray}
\label{solution_Riccatti}
\gamma^R_t&=&-\alpha\sigma^2+C_0\;\frac{C_1+C_2\,e^{-2C_0\sigma^{-2}t}}{C_1-C_2\,e^{-2C_0\sigma^{-2}t}}\\[1ex]
\nonumber
\end{eqnarray}
with $
C_0=\sigma\sqrt{\sigma^2\alpha^2+\beta^2},~
C_{1}=\nu_0+\alpha\sigma^2 + C_0$, $C_{2}=\nu_0+\alpha\sigma^2 - C_0$.\\[1ex]

\paragraph{Only expert opinions ($H=E$)}
If the investor's estimate on the drift is based only on  expert opinions  arriving at discrete points  $t_1,\ldots, t_N$ we have information $\mathbb{F}^E$.
For the conditional mean $\widehat \mu_t^E$ and the conditional variance $\gamma_t^E$ we have the following result.
\begin{lemma}\label{filter_E}\ \\[-2ex]
\begin{enumerate}
\item[(i)]
Between two information dates $t_k$ and $t_{k+1}$ it holds for $ t\in [t_k,t_{k+1})$, $k=0,\ldots,N-1$ that $\widehat \mu_t^E$ is  Gaussian  with
\begin{eqnarray}
\label{filter_E0}
\widehat \mu_t^E & =& e^{-\alpha(t-t_k)} \widehat \mu_{t_k}^E+\left(1- e^{-\alpha(t-t_k)}\right) \delta,\\
\gamma_t^E & = & e^{-2 \alpha(t-t_k)} \gamma_{t_k}^E  +\left (1-e^{-\alpha(t-t_k)}\right) \frac{\beta^2}{2\alpha}.
\label{cond_var_E0}
\end{eqnarray}
\item[(ii)]
At the information dates $t_k$
it holds that $\widehat \mu_{t_k}^E$ is  Gaussian  with
\begin{eqnarray}
\label{filter_E_update}
\widehat  \mu_{t_k}^E & =&  \lambda_k^E \;\widehat \mu_{t_k-}^E + (1- \lambda_k^E) \,Z_k \quad \text{where}\quad \lambda_k^E=\frac{\varexp_k}{\gamma^E_{t_k-}+\varexp_k}
\\[1ex]
\gamma_{t_k}^E & = & \lambda_k^E\,\gamma^E_{t_k-} = \frac{\gamma^E_{t_k-}\;\varexp_{k}}{\gamma^E_{t_k-}+\varexp_k}.
\label{cond_var_E_update}
\end{eqnarray}
For $t_0=0$ we set  $ \widehat \mu_{0-}^E:=m_0$ and $\gamma_{0-}^E:=\nu_0$.
\end{enumerate}
\end{lemma}
\begin{proof} 
Since the expert opinions arrive at discrete points in time it holds $\mathcal {F}^E_t=\mathcal {F}^E_{t_k} $  for $t\in [t_k,t_{k+1}),~k=0,\ldots,N-1$. Then we have
$\widehat \mu_t^E =E[\mu_t|\mathcal {F}^E_t] = E[\mu_t|\mathcal {F}^E_{t_k}] $ and
$\gamma_t^E =E[(\mu_t-\widehat \mu_t^E)^2|\mathcal {F}^E_t] = E[(\mu_t-\widehat \mu_t^E)^2|\mathcal {F}^E_{t_k}] $. According to  \eqref{mu_explicit} we get
\[\mu_t = \delta +e^{-\alpha (t-t_k)}\Big[(\mu_{t_k} -\delta) + \beta  \int_{t_k}^t e^{\alpha (s-t_k)} dB_s\Big].
\]
Therefore, $\widehat \mu_t^E = \delta + e^{-\alpha(t-t_k)}(E[ \mu_{t_k} \,|\,\mathcal{F}_t^E]-\delta)= e^{-\alpha(t-t_k)}\widehat \mu_{t_k}^E + (1-e^{-\alpha(t-t_k)})\delta $ and
\begin{eqnarray*}
\gamma_t^E &=& E\left[ \left(e^{-\alpha(t-t_k)}(\mu_{t_k}-\widehat \mu_{t_k}^E)
+ \beta\,e^{-\alpha(t-t_k)} \int_{t_k}^t e^{\alpha (s-t_k)} dB_s\right)^2\,\Big|\, \mathcal{F}_t^E \right]\\
&=&  e^{- 2 \alpha(t-t_k)} E\left[\left(\mu_{t_k}-\widehat \mu_{t_k}^E\right)^2  \mid \mathcal{F}_{t_k}^E \right]
 + \beta^2 \,e^{-2 \alpha(t-t_k)} E\left[ \int_{t_k}^t e^{ 2 \alpha (s-t_k)} ds\right]\\
&=& e^{- 2 \alpha(t-t_k)} \gamma_{t_k}^E + \frac{\beta^2}{2\alpha} \left(1- \,e^{-2 \alpha(t-t_k)}\right),
\end{eqnarray*}
where we used the martingale property of the stochastic integral and the It\^{o}-Isometry. This yields the representations in \eqref{filter_E0} and \eqref{cond_var_E0}.

The updating formulas \eqref{filter_E_update} and \eqref{cond_var_E_update} can be seen as an update of a degenerate discrete-time Kalman filter, see e.g.\ formulas (5.12) and (5.13) in Section 4.5 of Elliott, Aggoun and Moore \cite{Elliott et al. (1994)}. It is degenerate here, since there is no evolution in time from $t_{k}-$ to $t_k$. Alternatively the updating formulas may be computed directly as a Bayesian update of $\widehat\mu_{t_k-}^E$ given the $\mathcal{N}(\mu_{t_k}, \Gamma_k)$-distributed expert opinion, cf.\ Theorem II.8.2 in  Shiryaev \cite{Shiryaev}.
\end{proof}
\begin{remark}
The updating formula \eqref{filter_E_update} for the conditional mean $\widehat{\mu}^E_{t_k}$ shows that the filter after arrival of the $k$-th expert opinion is a weighted mean of the filter $\widehat{\mu}^E_{t_k-}$ before the arrival and the view $Z_k$ of the expert's view. The weight $\lambda_k^E\in [0,1]$ decreases with decreasing reliability $\varexp_k$ (i.e~increasing confidence)  of the expert. So more weight $1-\lambda_k^E$ is given to the view. For the limiting case $\varexp_k= 0$ (expert has full information) we have $1-\lambda_k^E=1$ and $\widehat{\mu}^E_{t_k}=Z_k=\mu_{t_k}$. For $\varexp_k= \infty$ we have $\lambda_k^E=1$ and $\widehat{\mu}^E_{t_k}=\widehat{\mu}^E_{t_k-}$, i.e., there is no impact of the expert's view since it carries no information on the unknown drift $\mu_t$.

From updating formula \eqref{cond_var_E_update} for the conditional variance $\gamma^E_{t_k}$ it can be seen that
$\gamma^E_{t_k} \le \min\{\gamma^E_{t_k-}, \varexp_k\}$, i.e.~the extra information never increases the conditional variance. For the limiting case $\varexp_k= 0$ we have $\gamma^E_{t_k}=0$ while for $\varexp_k= \infty$ we have $\gamma^E_{t_k}=\gamma^E_{t_k-}$. Again there is no impact of the expert's view.
\end{remark}

\paragraph{Return observations and expert opinions ($H=C$)}
This combination of the settings $H=R$ and $H=E$ is the case we are mainly interested in. An investor typically uses all available information,  stock returns and expert opinions.

\begin{lemma}
\label{filter_C}
\ \\[-2ex]
\begin{enumerate}
\item[(i)]
Between two information dates $t_k$ and $t_{k+1}$ it holds for $ t\in [t_k,t_{k+1})$, $k=0,\ldots,N-1$ that $\widehat \mu_t^C$ is  Gaussian  and satisfies
\begin{eqnarray}
\label{filter_C0}
d\widehat{\mu}^C_t &=& \big(\alpha\delta-(\alpha+\sigma^{-2}\gamma_t^C)\;\widehat{\mu}^C_t\big)
dt+\sigma^{-2}\gamma_t^C\;dR_t,
\\
\label{cond_var_C0}
\text{with }~~\gamma^C_t &=& -\alpha\sigma^2+C_0\;\frac{C_{1k}+C_{2k}\,e^{-2C_0\sigma^{-2}(t-t_k)}}{C_{1k}-C_{2k}\,e^{-2C_0\sigma^{-2}(t-t_k)}}
\end{eqnarray}
and initial values $\widehat\mu_{t_k}^C $ and $\gamma^C_{t_k}$, $k=0,\ldots,N-1$. The constant $C_0$ is given in  \eqref{solution_Riccatti} and for $k=0,\ldots,N-1$
\begin{equation}
\label{Ck}
C_{1k}:=\gamma_{t_k}^C+\alpha\sigma^2+ C_0\quad \text{and}\quad C_{2k}:=\gamma_{t_k}^C+\alpha\sigma^2- C_0 .
\end{equation}
\item[(ii)]
At the information dates $t_k$
 it holds that $\widehat \mu_{t_k}^C$ is  Gaussian  and
$\widehat \mu_{t_k}^C$ and $\gamma_{t_k}^C$ are obtained from the corresponding values at time $t_k-$ (before the arrival of the view) using the updating formulas  \eqref{filter_E_update} and \eqref{cond_var_E_update}, respectively, i.e.,
\begin{equation}\label{updatingC}
\widehat \mu_{t_k}^C  = \lambda_k^C\;\widehat \mu_{t_k-}^C + (1-\lambda_k^C)\,Z_k  \quad \text{ and }\quad
\gamma_{t_k}^C  = \lambda_k^C\,\gamma^C_{t_k-},
\end{equation}
where $\lambda_t^C = ({\gamma^C_{t_k-}+\varexp_k})^{-1}{\varexp_k}$ and $ \widehat \mu_{0-}^C:=m_0$,  $\gamma_{0-}^C:=\nu_0$.
\end{enumerate}
\end{lemma}
\begin{proof}
Between two information dates $t_k$ and $t_{k+1}$ we are in the standard situation of the Kalman filter with Gaussian initial values $\mu_{t_k}$, $\widehat \mu_{t_k}^C$ for  signal and filter and deterministic value $\gamma_{t_k}^C$ for the conditional variance. Since no additional expert opinions arrive in  $(t_k,t_{k+1})$ only the returns contribute to the investor filtration $\mathbb{F}^C$ and we have $\mathcal{F}_t^C= \mathcal{F}_{t_k}^C \vee \sigma\{R_s, t_k< s\le t\}$ for $ t\in (t_k,t_{k+1})$. So \eqref{filter_C0} and \eqref{cond_var_C0} follow immediately from the the corresponding Kalman filter equations \eqref{filter_R} and \eqref{solution_Riccatti}.

At the expert information dates we get the updating formulas in \eqref{updatingC} by applying a degenerate Kalman updating or  Bayesian updating formulas for Gaussian prior  $\widehat \mu_{t_k-}^C$ and Gaussian expert opinion $\expert_k$  as in the proof of Lemma \ref{filter_E}.
\end{proof}
\begin{remark}
\label{remark_EC}
We obtain $\widehat\mu_{t}^E$ and $\gamma^E_t$ from $\widehat\mu_{t}^C$ and $\gamma^C_t$ given in the above Lemma  for the limiting case $\sigma=\infty$. Then between the information dates $\widehat\mu_{t}^E$ is governed by the deterministic ODE $\frac{d}{dt}\widehat{\mu}^E_t = \alpha\big(\delta-\widehat{\mu}^E_t\big)$ while the conditional variance satisfies the linear ODE $\frac{d}{dt}\gamma^E_t=-2\alpha\gamma^E_t+\beta^2$. Solving these equations yields the expressions given in Lemma \ref{filter_E}. The interpretation of this limiting case $\sigma=\infty$ is that the volatility is such high that no additional information can be retrieved from observing the stock returns and thus it is enough to consider the expert opinions.
\end{remark}

\paragraph{Full information ($H=F$)}
For information $\mathbb{F}^F=\mathbb{G}$ it obviously holds $\widehat{\mu}^F_t=E[\mu_t|\mathcal{G}_t]=\mu_t$, i.e.~the conditional variance $\gamma_t^F$ is zero. Below we will study the conditional variances  $\gamma_t^E$ and  $\gamma_t^C$ and show, that these values tend to zero if the number of information dates $N$ tends to $\infty$, i.e.~asymptotically the value for full information is obtained.

\section{Properties of the Conditional Variance}
\label{cond_variance}

As a special feature of the filters using $\mathbb{F}^H$, $H\in \{R,E,C,F\}$, which we considered in Section \ref{partial_info},  we have a conditional variance $\gamma_t^H$  which is deterministic as it is known for the standard Kalman filter (case $H=R$). This leads to the following result for the second-order moment of the filter $\widehat \mu_t^H$ which will play a crucial role in the proof of our main result in Theorem \ref{opt_value_theorem}.
\begin{lemma}
\label{filter_moment_lemma}\ \\
For the second-order moment of the filter $\widehat \mu_t^H=E[\mu_t|\mathbb{F}^H]$, where $H\in \{R,E,C,F\}$, it holds for all $t\in[0,T]$
\begin{equation}
\label{filter_moment}
E[(\widehat \mu_t^H)^2] = E[\mu_t^2] - \gamma_t^H = \nu_t+ m_t^2 -\gamma_t^H.
\end{equation}
\end{lemma}
\begin{proof}
It holds
\begin{eqnarray*}
\gamma_t^H = E[(\mu_t-\widehat{\mu}^H_t)^2|\mathcal{F}^H_t] &=& E[\mu_t^2|\mathcal{F}^H_t] -2 E[\mu_t \widehat{\mu}^H_t|\mathcal{F}^H_t] + E[(\widehat{\mu}^H_t)^2|\mathcal{F}^H_t]\\
&=& E[\mu_t^2|\mathcal{F}^H_t]  - 2 (\widehat{\mu}^H_t)^2 + (\widehat{\mu}^H_t)^2
= E[\mu_t^2|\mathcal{F}^H_t]    - (\widehat{\mu}^H_t)^2.
\end{eqnarray*}
Since the conditional variance $\gamma_t^H$ for $H=R,E,C,F$ is deterministic, we have $E[\gamma_t^H] = \gamma_t^H$ and
\begin{eqnarray*}
E[(\widehat \mu_t^H)^2] &=& E\big[E[\mu_t^2|\mathcal{F}^H_t]\big] - E[\gamma_t^H]= E[\mu_t^2]   - \gamma_t^H
=  \nu_t+ m_t^2 -\gamma_t^H,
\end{eqnarray*}
where we have used that the drift $\mu$ is an Ornstein-Uhlenbeck process with mean $m_t$ and variance $\nu_t$ given in \eqref{mu_mean} and \eqref{mu_var}.
\end{proof}
The next proposition formally states an intuitive property of the filters namely that additional information on the unknown drift leads to an improvement of the drift estimate. This improvement can be measured by the conditional variance $\gamma^H$ of the filter $\widehat \mu_t^H$. We compare an investor  observing both returns and expert opinions (H=C) with an investor who  has access to only one of these  sources of information (H=R,E).
\begin{proposition}
\label{prop_cond_var}\ \\
It holds for all $t\in[0,T]$
\[\gamma_t^C \le \gamma_t^E \quad \text{and}\quad \gamma_t^C \le \gamma_t^R. \]
\end{proposition}
\begin{proof}
Between two information dates $t_k$ and $t_{k+1}$ the conditional variances $\gamma_t^H$ for $H=R,E,C$ satisfy the ODE
\[\frac{d}{dt}\gamma^H_t=f^H(\gamma_t^H),\quad \text{for } t\in[t_k,t_{k+1}),
\]
with initial value $\gamma^H_{t_k}$ where the r.h.s.~of this ODE is given by $f^R(y)=f^C(y)=-\sigma^{-2}y^2-2\alpha y+\beta^2$ (see Ricatti equation \eqref{cond_var_R} and Lemma \ref{filter_C}) and $f^E(y) =-2\alpha y+\beta^2  $ (see Remark \ref{remark_EC}).  It is well-known that this ODE has a unique solution.

For the proof of $\gamma_t^C \le \gamma_t^E$ we first note that $\gamma_0^C= \gamma_0^E = \nu_0 \frac{\varexp_0}{\nu_0+\varexp_0} $. It holds  $f^C(y) \le f^E(y)=f^C(y)+\sigma^{-2} y^2$. This implies that starting with coinciding initial values the solutions of the above ODE satisfy $\gamma_t^C \le \gamma_t^E$ on $[t_0,t_1)$.
This inequality also holds after the update at  $t=t_1$ since $\gamma_{t_1}^H =\gamma_{t_1-}^H\varexp_1/(\gamma_{t_1-}^H+\varexp_1)$ for $H=E,C$ and $x\mapsto x\varexp_1/(x+\varexp_1)$ is increasing in $x$. Iterating these arguments for $k=1,\ldots,N-1$ yields $\gamma_t^C \le \gamma_t^E$ for all $t\in[0,T]$.

For the proof of $\gamma_t^C \le \gamma_t^R$ we observe that $\gamma_0^C= \nu_0 \frac{\varexp_0}{\nu_0+\varexp_0} \le \nu_0 =\gamma_0^R  $ and $f^C=f^R$. The uniqueness of the solution of the above ODE yields that the inequality for the initial values is inherited to the solutions on $[t_0,t_1)$, i.e.~it holds $\gamma_t^C \le \gamma_t^R$. Then also  $\gamma_{t_1}^C$, the conditional variance  after the update at time $t_1$,  satisfies
\[\gamma_{t_1}^C =\gamma_{t_1-}^C\frac{\varexp_1}{\gamma_{t_1-}^C+\varexp_1} \le \gamma_{t_1-}^C\le  \gamma_{t_1}^R.\]
Iterating this argument for $k=1,\ldots,N-1$ yields $\gamma_t^C \le \gamma_t^R$ for all $t\in[0,T]$.
\end{proof}

The next Proposition formalizes another intuitive property of the filters. If the number $N$ of  expert opinions  tends to infinity, i.e., the extra information arrives more and more frequent, then in the limit for $N\to \infty$ we arrive at the case of full information about the unknown drift. This case is characterized by a vanishing conditional variance yielding in the limit a perfect estimate of $\mu_t$, see Remark \ref{vanishingCondVar} below.
\begin{proposition}
\label{prop_cond_var_asymN}\ \textbf{Asymptotics for $N\to \infty$}\\
Let $
\{t_0^{(N)},\ldots,t_N^{(N)}\}$ be a sequence of partitions of the interval $[0,T]$ into $N$ subintervals with mesh size
 $\Delta_N := \max_{k=1,\ldots,N} \{t_k^{(N)}-t_{k-1}^{(N)} \}$ and such that information dates are retained, i.e., $\{t_0^{(N)},\ldots,t_N^{(N)}\} \subseteq \{t_0^{(N')},\ldots,t_{N'}^{(N')}\}$ for $N'\ge N$. Moreover, let $(\varexp_k^{(N)})_{k=0,\ldots,N-1}$ be a sequence of corresponding variances of the expert opinions at time $t_k^{(N)}$. Assume that there  is some constant $\varexpbound>0$ such that $\varexp_k^{N} \le \varexpbound$ for all $k=0,\ldots,N-1$ and $N\in \N$.

Then it holds for the conditional variances $\gamma_t^{E,N}$ and $\gamma_t^{C,N}$, which correspond to these $N$ expert opinions,  that  for  all $t\in(0,T]$
\[\lim\limits_{N\to\infty,\Delta_N\to 0} \gamma_t^{E,N} = \lim\limits_{N\to\infty,\Delta_N\to 0} \gamma_t^{C,N} =0.\]
\end{proposition}
\begin{proof}
Since $0\le \gamma_t^C \le \gamma_t^E$, see
Proposition \ref{prop_cond_var}, we can restrict to the proof of  the assertion for $\gamma_t^{E,N}$. Moreover,
we  restrict to  expert opinions with constant  uncertainties $\varexp_k=\varexpbound$.
Then $\gamma_t^{E,N}$ dominates  the conditional variance   in the case where expert variances $\varexp_k$ are smaller than $\varexpbound$ and the general assertion follows.
We shall write $t_k$ for $t_k^{(N)}$ keeping the dependency on $N$ in mind.

For the dynamics of $\gamma_t^{E,N}$ we have from Lemma \ref{filter_E} for $k=0,\ldots,N-1$ and any $t\in  [t_k,t_{k+1})$
\begin{equation}\label{proof43eq1}
\gamma_t^{E,N}  =  e^{-2\alpha(t-t_k)}\gamma_{t_k}^{E,N}  +\left(1-e^{-2\alpha(t-t_k)}\right) \frac{\beta^2}{2\alpha}
\end{equation}
and
\begin{equation}\label{proof43eq2}
\gamma_{t_k}^{E,N} =\lambda_k^{E,N}\,\gamma^{E,N}_{t_k-}   \quad
\text{ with } \quad \lambda_k^{E,N} =
\frac{\varexpbound}{\gamma^{E,N}_{t_k-}+\varexpbound} \in(0,1].
\end{equation}
Since $1-e^{-2\alpha (t-t_k)} \le 2\alpha (t-t_k) \le 2\alpha \Delta_N$ it follows from \eqref{proof43eq1}, \eqref{proof43eq2}
\begin{equation}\label{proof43eq3}
\gamma^{E,N}_{t}\le  \lambda_k^{E,N} \gamma^{E,N}_{t_k-} + \beta^2 \Delta_N \qquad \text{for } ~t\in  [t_k,t_{k+1}).
\end{equation}
Iterating this inequality and denoting $\overline \lambda_k^{E,N} := \max\limits_{j=0,\ldots,k} \lambda_j^{E,N}$ yields for $\nu_0 = \gamma_{0-}^{E,N}$
\begin{equation}
\label{proof43eq4}
\gamma_t^{E,N} ~~\le~~ \left(\overline \lambda_k^{E,N}\right)^{k+1} \nu_0 + \Delta_N \beta^2  \sum\limits_{j=0}^{k} \left(\overline \lambda_k^{E,N}\right)^j
~~\le~~  \left(\overline \lambda_k^{E,N}\right)^{k+1} \nu_0 + \frac{\beta^2 \Delta_N}{1-\overline \lambda_k^{E,N}}
\end{equation}
for $t\in  [t_k,t_{k+1})$.

Now, let $u\in(0,T]$ and $\eps>0$. We have to show that we can choose $N$ such that $\gamma_u^{E,N} < \eps$.
By $k_N$ we denote the index for which $u\in[t_{k_N}, t_{k_N+1})$.

Suppose that for all $N_0$ there exists $N\ge N_0$ such that
\begin{equation} \label{proof43ass}
\min\{\gamma_{t_0-}^{E,N}, \ldots,  \gamma_{t_{k_N}-}^{E,N}\}\ge \eps/2.
\end{equation}
Then we would get
$\overline \lambda^{E,N}_{k_N} \le  (\eps/2+\varexpbound)^{-1}\,{\varexpbound}$ and thus by \eqref{proof43eq4}
with one iterations less
\begin{equation}\label{proof43eq6}
\gamma_{t_{k_N}-}^{E,N} \le \left(  \frac {\varexpbound}{\eps/2+\varexpbound}\right)^{k_N} +
\frac{\beta^2(\eps+2 \varexpbound)}{\eps} \, {\Delta_N}
\end{equation}
Since the bound for $\overline \lambda^{E,N}_{k_N}$ is strictly less than 1, independent of $N$ and $k_N$ is increasing in $N$ with $k_N\to \infty$ for $N\to\infty$, the right hand side of \eqref{proof43eq6} is decreasing and converges to $0$ for $N\to\infty$. In particular, we can choose $N_0$ such that $\gamma_{t_{k_N}-}^{E,N} < \eps/2$. But this is a contradiction to the assumption in \eqref{proof43ass}.

Therefore, there exists an $N_0$ such that for all $N\ge N_0$ there exists some index $l_N\le k_N$ with
$\gamma_{t_{l_N-}}^{E,N} < \eps/2$. For each $N\ge N_0$ we choose $l_N$ as the maximal index $l\le k_N$ for which $\gamma_{t_l-}^{E,N}< \eps/2$,  i.e.~$t_{l_N}$ is the last information time before (or equal) ~$t_{k_N}$ where the conditional variance before the update is smaller than $\eps/2$.
If $l_N=k_N$ then \eqref{proof43eq3} with $k=k_N, t=u$ implies that for $N$ large enough we have $\gamma_u^{E,N} \le \eps$ and the claim follows.\\
Otherwise, i.e.~if $l_N<k_N$ then we have for  $k= l_N+1, \ldots, k_N$ that $\gamma_{t_k}^{E,N} \ge \eps/2$ and thus $\lambda^{E,N}_{k} \le  (\eps/2+\varexpbound)^{-1}\,{\varexpbound}$ as above.
We choose $N_1\ge N_0$  such that $\beta^2(\eps+2 \varexpbound){\Delta_{N_1}}/{\eps} < \eps/2$.
An iteration as in \eqref{proof43eq4} starting with initial time $s=t_{l_N}$ and initial value $\gamma_{s}^{E,N}<\eps/2$ (instead of $0$ and $\nu_0$)  and using the upper bound $ (\eps/2+\varexpbound)^{-1}\,{\varexpbound}$ for $\lambda_k^{E,N}$ as in \eqref{proof43eq6} yields finally for all $N\ge N_1$
 \[
 \gamma_u^{E,N} \le  \left(\frac {\varexpbound}{\eps/2+\varexpbound}\right)^{k_N-l_N+1}\frac\eps{2}
 + \frac{\beta^2(\eps+2 \varexpbound)}{\eps} \, {\Delta_N} < \frac\eps{2} + \frac\eps{2} = \eps.
\]
\end{proof}

\begin{remark} \label{vanishingCondVar} Note that for full information we have $\gamma^F_t = 0$ for all $t\in(0,T]$. So Proposition \ref{prop_cond_var_asymN} shows that
$\gamma_t^{E,N}$ and $\gamma_t^{C,N}$ converge to $\gamma_t^F$ for increasing the number of expert opinions. In particular this shows that we gain full information in the limit. More precisely, by increasing the number of expert opinions we get an arbitrarily sharp estimate of $\mu_t$.
\end{remark}

For Proposition \ref{prop_cond_var_asym_t} we need for the existence of the limits some monotonicity properties of the conditional variances.

\begin{lemma}\label{monotonicity}
Between the informations dates, i.e.\ for $t\in [t_k, t_{k+1})$, we have
\begin{enumerate}
\item[(i)] for $H\in\{R, C\}$ that $\gamma_t^H$ is decreasing, if $\gamma_{t_k}^H > C_0-\alpha\,\sigma^2$, and increasing, if $\gamma_{t_k}^H < C_0-\alpha\,\sigma^2$, where $C_0$ is given in \eqref{solution_Riccatti},
\item[(i)] and for $H=E$ that $\gamma_t^E$ is decreasing, if $\gamma_t^E > \frac{\beta^2}{2\alpha}$, and increasing, if $\gamma_t^E < \frac{\beta^2}{2\alpha}$.
\end{enumerate}
\end{lemma}

\begin{proof}
The results follow for $H\in\{R, C\}$ from the dynamics by checking the sign of the right hand side in \eqref{cond_var_R} which also $\gamma_t^C$ satisfies and which led to \eqref{cond_var_C0}. For $H=E$ we take the derivative of \eqref{cond_var_E0} w.r.t. $t$ and check the sign.
\end{proof}

For the following proposition we consider an infinite time horizon $T=\infty$. Note that the filtering equations in Lemma \ref{filter_E} and Lemma \ref{filter_C} up to each $t$ remain valid. It turns out that the asymptotic bounds, which we derive for the conditional variance for an infinite time horizon and with equidistant information times, give quite accurate approximations also for the finite horizon case, see e.g.\ Figure \ref{cond_var_bounds}.
\begin{figure}[h]
\ifpdf \hspace*{-5mm}
\includegraphics[width=125mm,height=60mm]{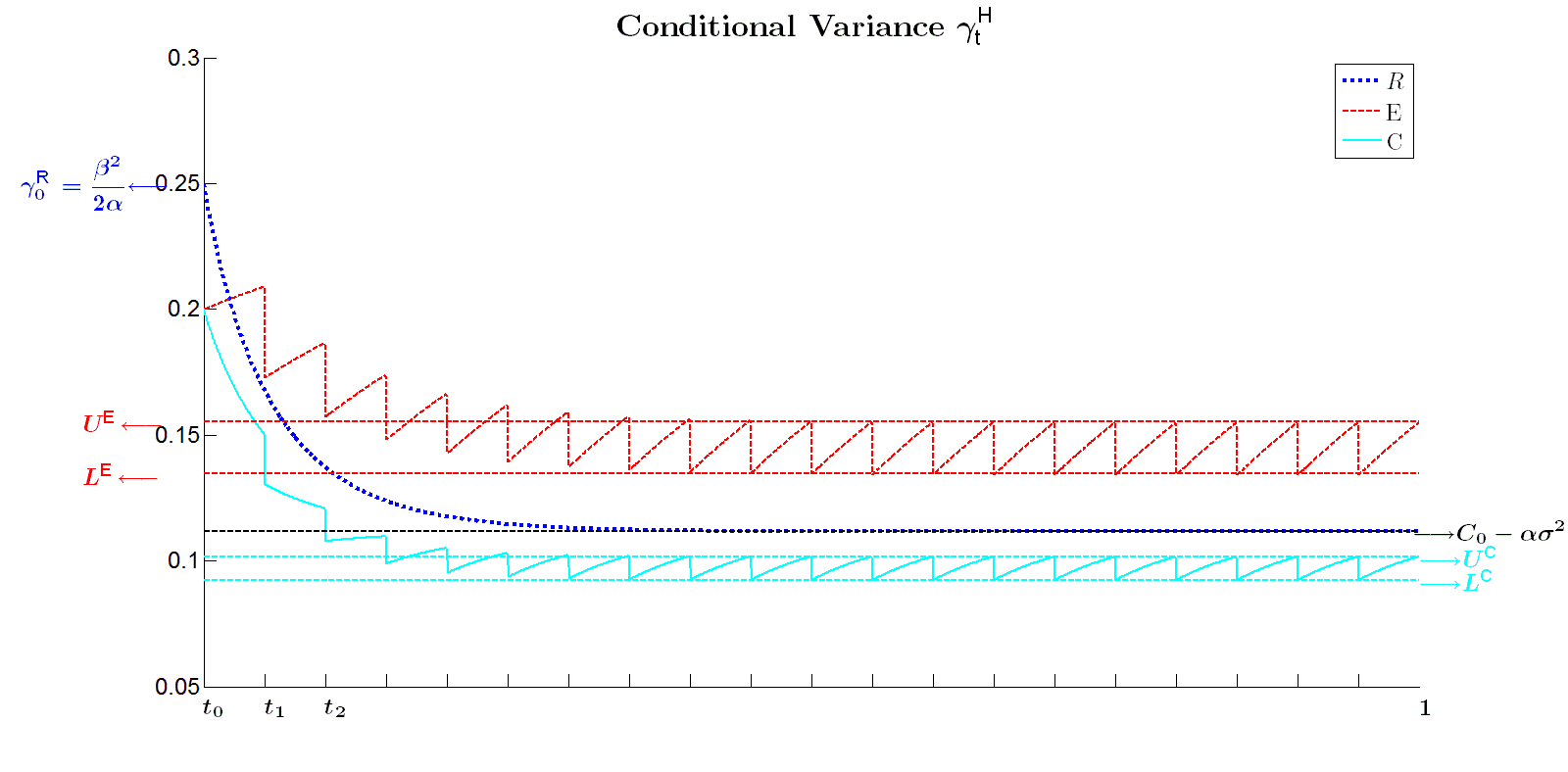}%
\vspace*{-4ex}
 \else \hspace*{-5mm}
\includegraphics[width=125mm,height=60mm]{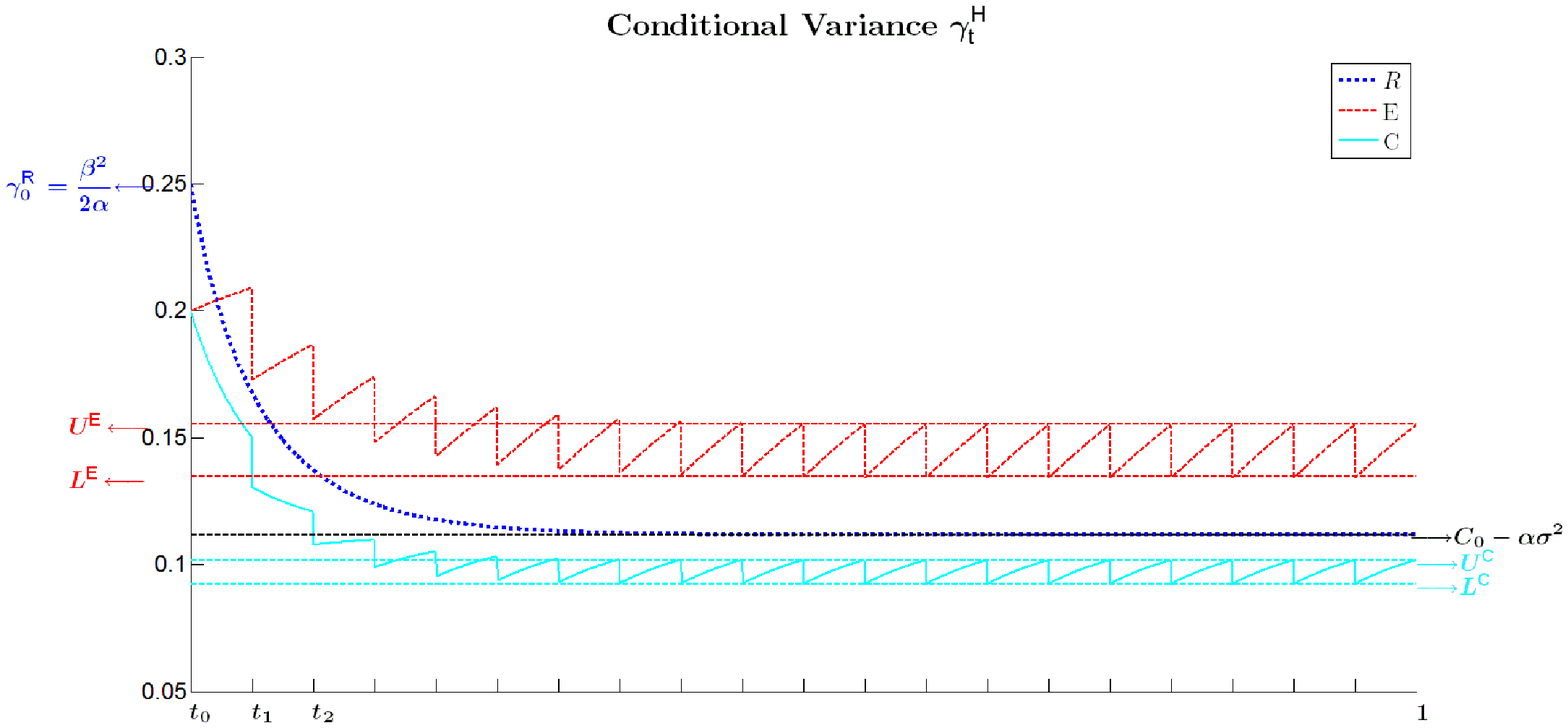}%

\vspace*{-4ex}
\fi
 \centering \caption{\label{cond_var_bounds}
 Asymptotic behaviour of the conditional variance $\gamma_t^H$ for $t\to\infty$, $H=R,E,C$\newline
 Parameters:  $\alpha=2, \beta=1, \sigma=0.15, \Gamma=1, N=20$
 }

\end{figure}

\begin{proposition}
\label{prop_cond_var_asym_t}\ \textbf{Asymptotics for $t\to \infty$}\\
Consider the model as above but with an infinite time horizon $T=\infty$ and assume that the expert opinions arrive at equidistant information dates $t_{k} = t_{k-1} + \Delta$ with some $\Delta>0$. First, without expert opinions ($H=R$) we have
\begin{equation}
\label{cond_var_asym_t_R}
\gamma^R_\infty  := \lim\limits_{t\to\infty} \gamma_t^{R} = C_0-\alpha\sigma^2,
\end{equation}
where $C_0=\sigma\sqrt{\sigma^2\alpha^2+\beta^2}$ as in \eqref{solution_Riccatti}.\\
For $H=E,C$ let $t_k=k\Delta$  and $\varexp_k=\varexp>0$ for $k=0,1,\ldots$.
Then it holds for the conditional variances $\gamma_t^{H}$
\begin{equation}
\label{cond_var_asym_t_EC}
\lim\sup\limits_{t\to\infty} \gamma_t^{H} = U^{H} \quad \text{and}\quad
\lim\inf\limits_{t\to\infty} \gamma_t^{H} = L^{H},
\end{equation}
\begin{eqnarray*}
\text{where}\quad U^H &=& \frac{1}{2a^H}\Big(-b^H + \sqrt{(b^H)^2-4a^Hc^H}\;\Big)\quad \text{and}\quad L^H = \frac{\varexp U^H }{\varexp+ U^H}\quad \text{with}
 \end{eqnarray*}
 \[\begin{array}{rlrl}
 a^E &= 1, &        a^C&= \frac{1}{2\alpha\sigma^2}\big((1-d^C)(\varexp+\alpha\sigma^2) +(1+d^C)C_0\big),\\[1ex]
 b^H &= -(1-d^H)\Big( \frac{\beta^2}{2\alpha}-\varexp\Big),&
 c^H &= -(1-d^H) \frac{\beta^2}{2\alpha}\varexp, \quad H\in \{E,C\},\\[1ex]
  d^E &= \exp(-2\alpha \Delta), & d^C &= \exp(-2C_0\sigma^{-2}\Delta).
\end{array}
\]
\end{proposition}
\begin{proof}
The proof of \eqref{cond_var_asym_t_R} follows immediately from representation \eqref{solution_Riccatti} for $\gamma^R_t$.

For the proof of assertion \eqref{cond_var_asym_t_EC} we observe that
there exists some index $k_0$ such that for $k\ge k_0$ the conditional variance $\gamma_t^H$ is increasing between two information dates $t_k$ and $t_{k+1}$. To prove this, note that
\[
\lambda^E_k\le \Lambda^E:= \frac{\Gamma}{\beta^2+\Gamma} <1 \quad \mbox{ if } \quad \gamma_{t_k-}^E \ge \beta^2
\]
and
\[
\lambda^C_k\le \Lambda^C := \frac{\Gamma}{C_0-\alpha \sigma^2 +\Gamma} <1 \mbox{ if } \gamma_{t_k-}^C \ge C_0-\alpha \sigma^2.
\]
Since by Lemma \ref{monotonicity} $\gamma_t^E$ and $\gamma_t^C$ are decreasing on $[t_k, t_{k+1})$ as long as they lie above these boundaries $\beta^2$ and $C_0-\alpha \sigma^2$, respectively, iterating the updating formulas in Lemma \ref{filter_E} and Lemma \ref{filter_C} yields
\[
\gamma_{t_k}^H = \gamma_{t_k-}^H \lambda_k^H \le \gamma_{t_{k-1}}^H \Lambda^H \le \ldots \le \nu_0 \left(\Lambda^H\right)^{k+1}, \quad H= E, C.
\]
Since $\Lambda^H <1$, $\gamma_t^H$ finally falls below the corresponding bound and with a similar argument as above then stays below this boundary. By Lemma \ref{monotonicity}, $\gamma_t^H$ is increasing between the information dates below this boundary and thus a $k_0$ as stated above can be found in both cases.

Moreover, $\gamma_t^H$ is bounded from below by $0$,  and hence for $t\ge t_{k_0}$ we have $\gamma_t^E\in[0,\beta^2]$ and $\gamma_t^C\in[0,C_0- \alpha \sigma^2]$. One can further show that the sequences $(\gamma^H_{t_k-})_k$ and $(\gamma^H_{t_k})_k$ are either decreasing  or increasing (the latter when starting with small $\nu_0$). Therefore, the limits exist and we have
$$U^H= \lim\sup\limits_{t\to\infty} \gamma_t^{H} = \lim\limits_{k\to\infty} \gamma_{t_k-}^{H}\quad \text{and}\quad L^H= \lim\inf\limits_{t\to\infty} \gamma_t^{H} = \lim\limits_{k\to\infty} \gamma_{t_k}^{H}.$$

For $H=E$, Lemma \ref{filter_E} yields for the conditional variances before the update at the information dates $t_{k+1}$,    $k=0,1,\ldots$,
\[\gamma_{t_{k+1}-}^E  =  G^E(\gamma_{t_{k}}^E)\quad  \text{with }\quad  G^E(x) :=d^E \;x  +(1-d^E)\;\frac{\beta^2}{2\alpha}\]
where $d^E=e^{-2\alpha \Delta}$, and for $H=C$ Lemma \ref{filter_C} yields
\[\gamma^C_{t_{k+1}-} = -\alpha\sigma^2+C_0\;\frac{C_{1k}+C_{2k}\,e^{-2C_0\sigma^{-2}\Delta}}{C_{1k}-C_{2k}\,e^{-2C_0\sigma^{-2}\Delta}} \quad \text{where}\quad C_{1/2,k}=\gamma_{t_k}^C+\alpha\sigma^2\pm C_0.\]
Hence, $\gamma^C_{t_{k+1}-}$ can be expressed as
\[
\gamma^C_{t_{k+1}-} = G^C(\gamma_{t_{k}}^C)
~\text{where }~~ G^C(x)  :=  -\alpha\sigma^2 + C_0\frac{(1+d^C)( x  + \alpha\sigma^2) + C_0(1- d^C)}{(1-d^C) (x + \alpha\sigma^2) + C_0(1+ d^C)}
\]
and $d^C=e^{-2C_0\sigma^{-2}\Delta}$.
Since the limits $U^H$ and $L^H$ exist, we can substitute in the above equations for $\gamma_{t_{k+1}-}^H$ and in the updating formula $\gamma_{t_k}^H  = \lambda^H_k \gamma^H_{t_k-} = ({\gamma^H_{t_k-}+\varexp})^{-1} \gamma^H_{t_k-}\,\varexp$, given in \eqref{cond_var_E_update} and \eqref{updatingC}, first $U^H$ for $\gamma^H_{t_k-}$ and $\gamma^H_{t_{k+1}-}$ and second $L^H$ for $\gamma^H_{t_k}$ to compute the limits. Therefore,
 $U^H$ and $L^H$ satisfy
\[U^H=G^H(L^H) \quad \text{ and }\quad L^H= \frac{U^H\varexp}{U^H+\varexp}.\]
Substituting the second into the first equation yields after some algebra the  quadratic equation
$a^H(U^H)^2 +b^H U^H + c^H = 0 $ for $U^H$
with coefficients $a^H,b^H,c^H$  given in the proposition.
Since $0<d^H<1$ it holds  $a^H>0$ and $c^H<0$ and  $|b^H|<((b^H)^2-4a^Hc^H)^{1/2}$, hence there is  one negative and one positive real solution. We are only interested in the positive solution which is given by $U^H=\frac{1}{2a^H}(-b^H + ((b^H)^2-4a^Hc^C)^{1/2})$ yielding the expression in the proposition. The expression for $L^H$ follows from the updating formula.
\end{proof}

A detailed look at the formulas for $U^H$ and $L^H$ in Proposition \ref{prop_cond_var_asym_t} reveals:

\begin{corollary}
Under the conditions of Proposition  \ref{prop_cond_var_asym_t} for $H\in\{E,C\}$ the limits
$U^H$ and $L^H$ are decreasing in $\Delta$ with
\[
\lim_{\Delta\to 0} U^H = \lim_{\Delta \to 0}  L^H = 0.
\]
Further, $U^E\le \frac{\beta^2}{2\alpha}$ and $U^C\le \gamma^R_\infty$.
\end{corollary}

\section{Portfolio Optimization Problem}
\label{optimization}

Now that we have the filtering results at hand, it is quite straightforward to compute the optimal strategy and explicit representations for the value functions for our four cases of $H$ for logarithmic utility. This illustrates the influence of the expert opinions.

We describe the
self-financing trading of an investor by the initial capital $x_0>0$
and the  $\mathbb{F}$-adapted trading strategy $\pi=(\pi_t)_{t\in[0,T]} $
where $\pi_t\in\R$ represents  the proportion of wealth invested in stocks  at
time $t$.  It is well-known that in this setting the wealth process
$X^{\pi}$ of the portfolio has the  dynamics
\begin{eqnarray} \label{wealth_phys}
\frac{dX_t^{\pi}}{X_t^{\pi}}= \pi_t\frac{dS_t}{S_t} &= & \pi_t\mu_t
dt+\pi_t\sigma dW_t,\quad X_0^{\pi}=x_0.
\end{eqnarray}
We denote by
\begin{equation*}
\mathcal{A}^H=\{\pi= (\pi_t)_{t\in[0,T]}  \colon \text{
$\pi$ is $\mathbb{F}^H$-adapted, } X^\pi_t \ge 0,  t\in[0,T],  E[\int_0^T \pi_t^2\, dt ]<\infty \}
\end{equation*}%
the  class of {\em admissible trading strategies}, where $H\in \{R,E,C,F\}$.

We assume that the investor wants to maximize the expected logarithmic utility
of terminal wealth. The optimization problem thus reads
\begin{eqnarray}
\label{opti_org}
V^H(x_0) =
\sup\{ E[\log (X_T^{\pi})] \colon {\pi\in\mathcal{A}^H}\},
\end{eqnarray}
where $V^H(x_0)$ is called the value of the optimization problem for given initial capital $x_0$.
This is a maximization problem under partial information since  we
have required that the strategy $\pi$ is adapted to the investor
filtration  $\mathbb{F}$. In particular we are interested in an optimal strategy $\pi^*=\pi^{*,H}$ which attains the optimal value, i.e., $V^H(x_0) = E[\log(X_T^{\pi^*})]$.

\begin{proposition} \label{lemma:h_opti_log}
The optimal strategy for problem (\ref{opti_org}) is
\[
\pi_t^* = \sigma^{-2}\,\widehat{\mu}_t^H, \quad t\in[0,T].
\]
\end{proposition}

\begin{proof}
From \eqref{wealth_phys} it follows that
\begin{eqnarray}
\nonumber \log X_T^\pi &=& \log x_0 +
\int_0^T\Big(\pi_t\,{\mu}_t- \frac{1}{2} (\sigma
\pi_t)^2 \Big)dt + \int_0^T \pi_t\sigma d{W}_t.
\end{eqnarray}
For $\pi\in \mathcal{A}^H$ we have $E[\int_0^T \pi_t^2 \sigma^2dt]<\infty$, hence the latter integral is a martingale, in particular $E[\int_0^T \pi_t \sigma d{W}_t]=0$. Therefore, using a Fubini argument, the tower property of conditional expectations and that  $\pi_t$ is $\mathcal{F}_t^H$-measurable, we get
\begin{eqnarray}\nonumber
  E[\log X_T^\pi] &=&  \log x_0 + E\left[\int_0^T\left(\pi_t\,{\mu}_t-
\frac{1}{2}( \sigma \pi_t)^2\right)dt \right]\\\nonumber
&=&
  \log x_0 + \int_0^T E\left[ E\left[\pi_t\,{\mu}_t-
\frac{1}{2}( \sigma \pi_t)^2 \,\Big|\, \mathcal{F}^H_t \right] \right]dt \\\nonumber
&=&
  \log x_0 + \int_0^T E\left[\pi_t\,E\left[ {\mu}_t \mid \mathcal{F}^H_t \right] -
\frac{1}{2}( \sigma \pi_t)^2 \right]dt \\
&=&
 \log x_0 + \int_0^T E\left[\pi_t\,\widehat{\mu}_t^H-
\frac{1}{2}( \sigma \pi_t)^2 \right] dt.
\label{hint}
\end{eqnarray}
From \eqref{filter_moment} we have $E[(\widehat{\mu}_t)^2]\le \nu_t+m_t^2  <\infty$,
where $m_t$ and $\nu_t$ are bounded, hence $E\Big[\int_0^T (\widehat{\mu_t^H})^2\, dt \Big]<\infty$ and  the stated strategy $\pi^*$  is indeed admissible. Moreover, for
all $t\in[0,T]$ the quantity $\pi^*_t$ maximizes the integrand in
(\ref{hint}) pointwise, which implies that $\pi^*$ is the maximizer of
$E[\log(X_T^\pi)]$.
\end{proof}

\begin{remark}
\label{interpret_log}  If  the drift $\mu_t$ is observable, then the
optimal strategy  is well-known: at time $t$ one has to invest the
fractions  $\sigma^{-2} \mu_t$ of wealth in the risky stocks.
So for logarithmic utility  the so-called {\em certainty equivalence
principle} holds, i.e.~the optimal strategy under partial
information is obtained by replacing the  unknown drift $\mu_t$ by
the filter estimate $\widehat{\mu}_t^H$ in the formula for the optimal
strategy under full information. Note that this principle is no longer valid for other utility functions, e.g.~for power utility, see Brendle \cite{Brendle (2006)}, Lakner \cite{Lakner (1998)},
Sass and Haussmann \cite{Sass and Haussmann (2004)}).
\end{remark}

Now we can state our main result which provides closed form expressions for the optimal values $V^H(x_0)$ of the considered utility maximization problem under partial information.
\begin{theorem}
\label{opt_value_theorem}
For the optimal value of the optimization problem \eqref{opti_org} it holds
\begin{equation}
\label{opt_value}
V^H(x_0) = \log x_0 + \frac1{2\, \sigma^2}\left(A(m_0,\nu_0) - B^H\right) \quad \text{ for } \quad H\in\{R,E,C,F\},
\end{equation}
where
\[
A(m_0,\nu_0) :=   \int_0^T E[\mu_t^2]\, dt
\quad \text{ and }\quad B^H:=  \int_0^T \gamma_t^H\, dt.
\]
In particular, it holds
\begin{eqnarray}
 \label{opt_value_A}
 A(m_0,\nu_0)&=&
  \Big(\delta^2 + \frac{\beta^2}{2\alpha}\Big)T +
 2\delta(m_0-\delta)\frac{1-e^{-\alpha T}}{\alpha}\\
 \nonumber
 && \hspace*{10mm}
+ \Big((m_0-\delta)^2+\nu_0-\frac{\beta^2}{2\alpha}\Big)\frac{1-e^{-2\alpha T}}{2\alpha}
\end{eqnarray}
and the values of $B^H$ are
\begin{eqnarray}
\nonumber
B^R &=& (C_0-\alpha\sigma^2)T+\sigma^2\log \frac{C_1-C_2\, e^{-2C_0 \sigma^{-2}T}}{2C_0},\\
\nonumber
B^E &=&  \frac{\beta^2}{2\alpha}T- \left(\frac{1-e^{-2\alpha(t_{k+1}-t_k)}}{2\alpha}\right)
 \sum\limits_{k=0}^{N-1} \left(\frac{\beta^2}{2\alpha}-\gamma^E_{t_k}\right),\\
\nonumber
B^C &=&  (C_0-\alpha\sigma^2)T+\sigma^2\sum\limits_{k=0}^{N-1} \log \frac{C_{1k}-C_{2k} e^{-2C_0 \sigma^{-2}(t_{k+1}-t_k)}}{2C_0},\\
\nonumber
 B^F  &=& 0,
\end{eqnarray}
where  $C_0,C_1,C_2$ are given in  \eqref{solution_Riccatti} and  $C_{1k},C_{2k}$ in \eqref{Ck}.
\end{theorem}

\begin{proof} Substituting the optimal strategy $ \pi^*
= \sigma^{-2}\,\widehat{\mu}^H$ given in Proposition \ref{lemma:h_opti_log} into the expression for $E[\log X_T^\pi]$ in \eqref{hint} yields
\begin{eqnarray*}
V^H(x_0) = E[\log X_T^{\pi^*}] &=& \log x_0 + E\Big[\int_0^T\Big(\pi_t^*\,\widehat{\mu}_t^H-
\frac{1}{2}( \sigma \pi_t^*)^2\Big)dt \Big]\\
 &=& \log x_0 + \frac{1}{2\sigma^2}\int_0^T E\big[(\widehat{\mu}_t^H)^2\big] \,dt \\
 &=& \log x_0 +  \frac{1}{2\sigma^2}\int_0^T (E[\mu_t^2] -\gamma_t^H) \,dt
\end{eqnarray*}
where we have used the expression for $E[(\widehat \mu_t^H)^2]$ given in Lemma \ref{filter_moment_lemma}.
Evaluating the integral $\int_0^TE[\mu_t^2]\, dt = \int_0^T(\nu_t+ m_t^2)\, dt $ using the expressions for the mean $m_t$ and variance $\nu_t$ of the drift given in \eqref{mu_mean} and \eqref{mu_var} yields the expression \eqref{opt_value_A} for $A(m_0,\nu_0)$.
Finally, evaluating the integral $\int_0^T\gamma_t^H\, dt$ using the expressions for $\gamma_t^R$ given in \eqref{solution_Riccatti}, for $\gamma_t^E$ and $\gamma_t^C$ given in Lemma \ref{filter_E} and \ref{filter_C} and $\gamma_t^F=0$ yields the formulas for $B^H$ in the theorem.
\end{proof}

\begin{remark}\ \\[-3ex]
\begin{enumerate}
\item
Initializing the filter  with the stationary distribution of the drift, i.e.~$\mu_0\sim \mathcal{N}(\delta,\frac{\beta^2}{2\alpha})$, simplifies
the expression for $A$  in Theorem \ref{opt_value_theorem} yielding
\[ A(m_0,\nu_0)=A\Big(\delta,\frac{\beta^2}{2\alpha}\Big) = \Big(\delta^2 + \frac{\beta^2}{2\alpha}\Big)T. \]
\item
The well-known result for the classical Merton problem with constant drift $\mu_t\equiv   m_0$ can be  obtained as special case by setting $\beta=\nu_0=0$ and $\delta=m_0$. Then the optimal strategy is $\pi^*=m_0/\sigma^{2}$ and the optimal  value is simply
$V(x_0)=\log x_0+\frac{1}{2\sigma^2}m_0^2 \, T$.
\end{enumerate}
\end{remark}
\paragraph{Properties of the value function} We can now use Theorem \ref{opt_value_theorem} and the properties of $\gamma_t^H$ derived in Section \ref{cond_variance} to compare the value functions.

\begin{corollary}
\label{prop_value}\ \\
It holds
\[
\max\{V^E(x_0), V^R(x_0)\} \le V^C(x_0) \le V^F(x_0). \]
\end{corollary}
\begin{proof} From Theorem \ref{opt_value_theorem} we have the representation
\[V(x_0) = \log x_0 + \frac{1}{2\sigma^2}\,A(m_0,\nu_0) -\frac{1}{2\sigma^2} \int_0^T  \gamma_t\, dt.\]
The inequality $\min\{ \gamma_t^E,\gamma_t^R\} \ge  \gamma_t^C \ge \gamma_t^F=0$ given in Proposition \ref{prop_cond_var} which holds  for all $t\in[0,T]$ yields the  assertion.

\end{proof}

\begin{corollary}
\label{prop_value_asymN}\ \textbf{Asymptotics for $N\to \infty$}\\
Under the assumptions of Proposition \ref{prop_cond_var_asymN} and denoting the value functions corresponding to $N$ expert opinions as specified in that proposition by $V^{E,N}(x_0)$ and $V^{C,N}(x_0)$,
we get
\[\lim\limits_{N\to\infty} V^{E,N}(x_0) = \lim\limits_{N\to\infty} V^{C,N}(x_0) =V^{F}(x_0).\]
\end{corollary}
\begin{proof} The proof is analogous to the proof of Corollary \ref{prop_value} and uses the asymptotic properties of $\gamma_t^{E,N}$ and $\gamma_t^{C,N}$ given in Proposition \ref{prop_cond_var_asymN}.
\end{proof}

\section{Numerical Results}
\label{numerics}
In this section we illustrate the findings of the previous sections. The numerical
experiments are based on a financial market model where the drift $\mu$ follows an Ornstein-Uhlenbeck process as given in \eqref{drift} and \eqref{mu_explicit}, volatility $\sigma$ is constant and the interest
rate  equals zero.
For simulated drift process, stock prices and expert opinions we consider the maximization of
expected logarithmic utility of terminal  wealth.
We assume that $N$   expert opinions  with normally distributed views $Z_k$ arrive at equidistant information dates $t_k=kT/N$ and that their variances are constant, i.e.~$\varexp_k=\varexp>0,~k=0,\ldots,N-1,~N\in\N$. If not stated otherwise we use the model parametes given in Table \ref{parameter} in the subsequent simulations.

\begin{table}[h]
\begin{tabular}{|ll|r||lll|r|}
\hline
\rule{0mm}{2.5ex}%
Investment horizon & $T$ & $1$ year          & Drift: & Mean & $\delta$ & $0.05$
\\\hline \rule{0mm}{2.5ex}%
Stock volatility & $\sigma$ & $0.25$         &       & Mean reversion speed & $\alpha$ & $3$
\\\hline \rule{0mm}{2.5ex}%
Expert's variance &$\Gamma$ &$0.5^2$&       & Volatility & $\beta$ & $1$\\
\hline
\end{tabular}
\\[1ex]
 \centering \caption{\label{parameter}
 \small Model paramters
 }
\end{table}

\paragraph{Filter}
Figure \ref{drift} shows the filter $\widehat{\mu}^H$ and the conditional variance $\gamma^H$ for $H=R,E,C$. Note that for $H=F$ (full information) the filter $\widehat{\mu}^F$ coincides with the drift process $\mu$ while $\gamma^F=0$.
The upper panel shows the  simulated path of the return process $R$ from which the filter $\widehat{\mu}^R$ is computed. The plot also shows the path of $\int_0^t \mu_s\,ds$ which would be the return process for $\sigma=0$.
In the second  panel the initial value $\mu_0$ is assumed to be unknown and its distribution is the stationary distribution. Here the initial value of the filters  is $\widehat{\mu}_{0-}^H=E[\mu_0]=\delta$. The conditional variance, shown in the lower left panel,  starts with the stationary variance, i.e.~$\gamma_{0-}^H=\nu_0=\beta^2/(2\alpha)$. Note that for $H=E,C$ the first update is at time $t=0$ and we have $\widehat{\mu}_0^E=\widehat{\mu}_0^C\neq \widehat{\mu}_0^R$ and $\gamma_0^E=\gamma_0^C< \gamma_0^R$.  In the third panel  we assume  a known (deterministic) initial value $\mu_0=m_0$. So at time $t=0$  also for $H=E,R,C$  we have full information on the drift and it holds $\nu_0=\gamma_0^H=0$ and $\widehat{\mu}^H_0=m_0$.

For emphasizing the effect of the filter updates  due to the expert opinions  we have chosen $\Gamma=0.2^2$, which is  quite small and corresponds to a very reliable expert. In the second and third panel the expert views $Z_k$ are marked by red crosses. It can be seen that for $H=E,C$  the conditional variance $\gamma_t^H$ jumps down in the information dates and the filter $\widehat{\mu}^H_t$ is quite close to the actual value of the drift $\mu_t$. Between the information dates $\gamma^H$ increases and the filter $\widehat{\mu}^H$ is driven back to its mean  $\delta$.
While the conditional variances $\gamma^H$ start for unknown resp. known initial value with different initial values the lower panel shows that
$\gamma^R$ quickly reaches the asymptotic value  $\gamma^R_\infty=C_0-\alpha\sigma^2$ and also  the asymptotic upper and lower bounds for $\gamma^E$ and $\gamma^C$ from Proposition \ref{prop_cond_var_asym_t} apply even for small times.

\begin{figure}[p]

\ifpdf \hspace*{-10mm}
\includegraphics[width=140mm,height=110mm]{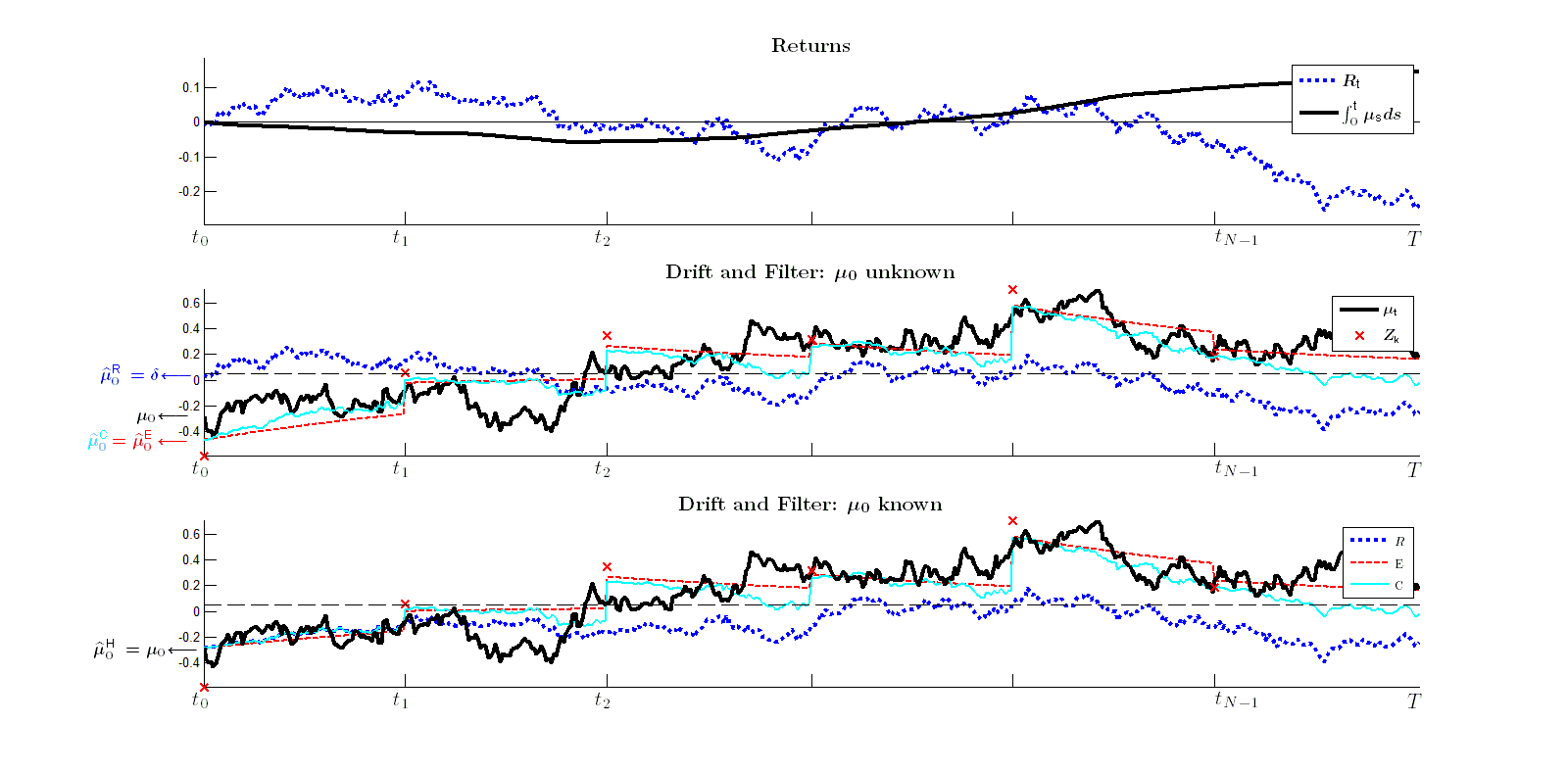}%

\hspace*{-5mm}
\includegraphics[width=120mm,height=40mm]{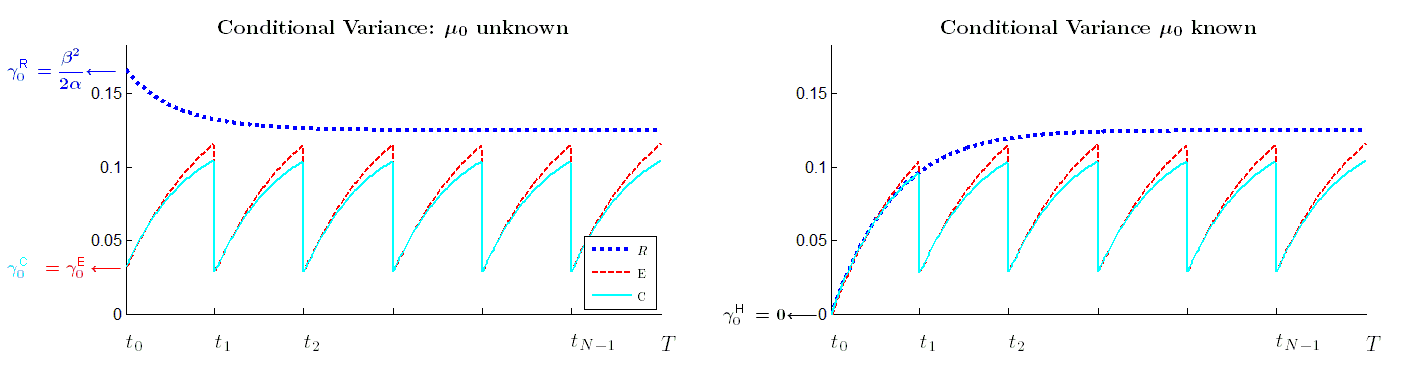}%
\vspace*{-2ex}
 \else \hspace*{-2mm}
\includegraphics[width=120mm,height=110mm]{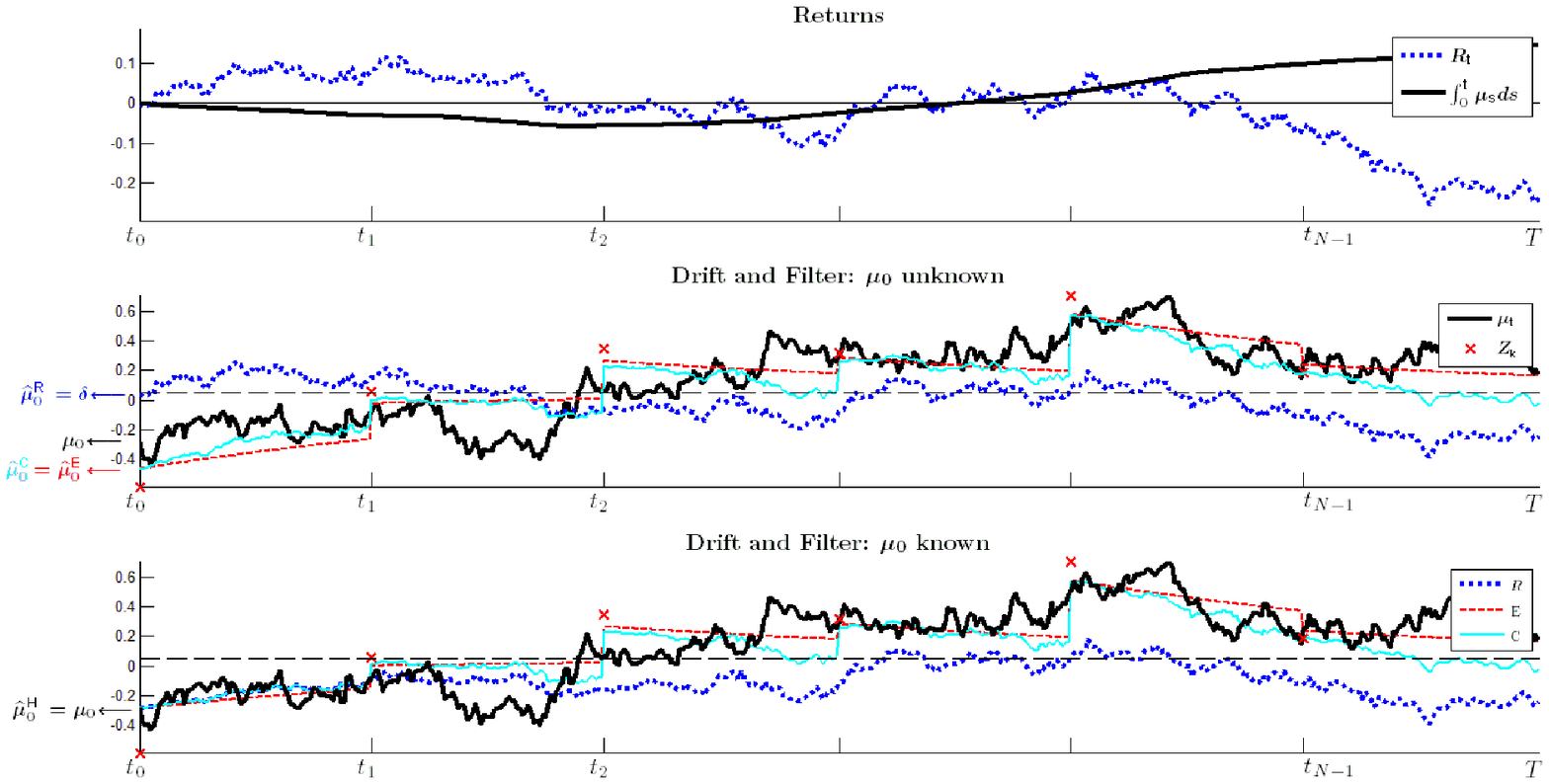}%

\hspace*{-5mm}
\includegraphics[width=120mm,height=40mm]{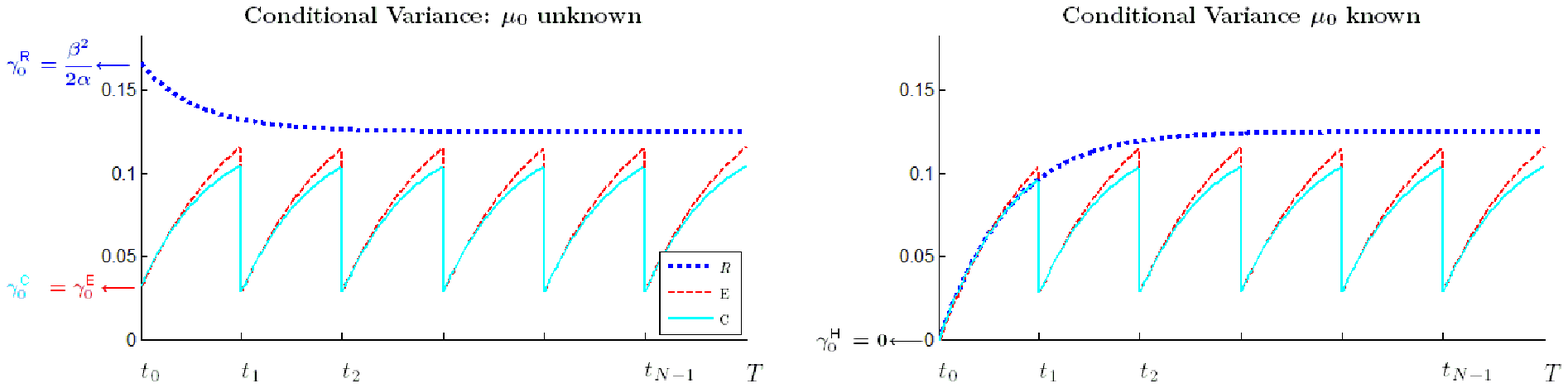}%
\fi

 \centering \caption{\label{figdrift1}
 \small
 Return $R$, Drift $\mu$, filters $\widehat{\mu}^H$ and conditional variances~$\gamma^H$
 \newline
 1.~panel: return $R$ and $\int_0^t \mu_s\,ds$ (return for $\sigma=0$)\newline
 2.~panel: drift $\mu$ and filters $\widehat{\mu}^H$ for unknown initial value
 $\mu_0\sim \mathcal{N}(\delta, \frac{\beta^2}{2\alpha})$
 \newline
 3.~panel: drift $\mu$ and filters $\widehat{\mu}^H$ for known initial value $\mu_0=m_0, \nu_0=0 $
 \newline
 4.~panel: conditional variances $\gamma^H$ for unknown/known (left/right)~$\mu_0$\newline
 Parameters as in Table \ref{parameter},  $\Gamma=0.2^2, N=6$
 }
\end{figure}

\paragraph{Efficiency}
In order to quantify the monetary value
of information contained in the  observations of stock returns and expert opinions and  in particular the value of the additional information due to the expert opinions, we compare four investors maximizing their expected log-utility from terminal wealth. First, the ``fully informed'' or $F$--investor can  observe the drift. Second,
the $E$--investor   has only access to expert opinions while the $R$--investor only observes stock returns. Finally the $C$--investor has access to (the combination of) stock returns and expert opinions.  Now we consider for $H=E,R,C$
the initial capital $x_0^H$  which the $H$--investor needs to obtain the same maximized
expected log-utility at time $T$ as the fully informed investor who started at time $0$ with unit
wealth $x_0^F=1$.
The difference  $x_0^H-x_0^F$ can be interpreted as loss of information
for the (non fully informed) $H$-investor while the ratio
\[\varrho^H = \frac{x_0^F}{x_0^H} = \frac{1}{x_0^H}\]
is a measure for the efficiency of the $H$-investor.

The initial capital required
by the $H$--investor $x_0^H$  is obtained as solution of the equation $V^H(x_0^H)=V^F(1)$.
From Theorem  \ref{opt_value_theorem} it follows $V^H(x_0^H) =\log x_0^H + \frac{1}{2\sigma^2}(A(m_0,\nu_0) - B^H)$ and $V^F(1)=  \frac{1}{2\sigma^2} A(m_0,\nu_0)$ since $B^F=0$, hence
\[x_0^H = \exp\Big(\frac{B^H}{2\sigma^2}\Big) \quad\text{and}\quad \varrho^H=\exp\Big(-\frac{B^H}{2\sigma^2}\Big).\]

\begin{figure}[h]
\ifpdf \hspace*{-10mm}
\includegraphics[width=140mm,height=70mm]{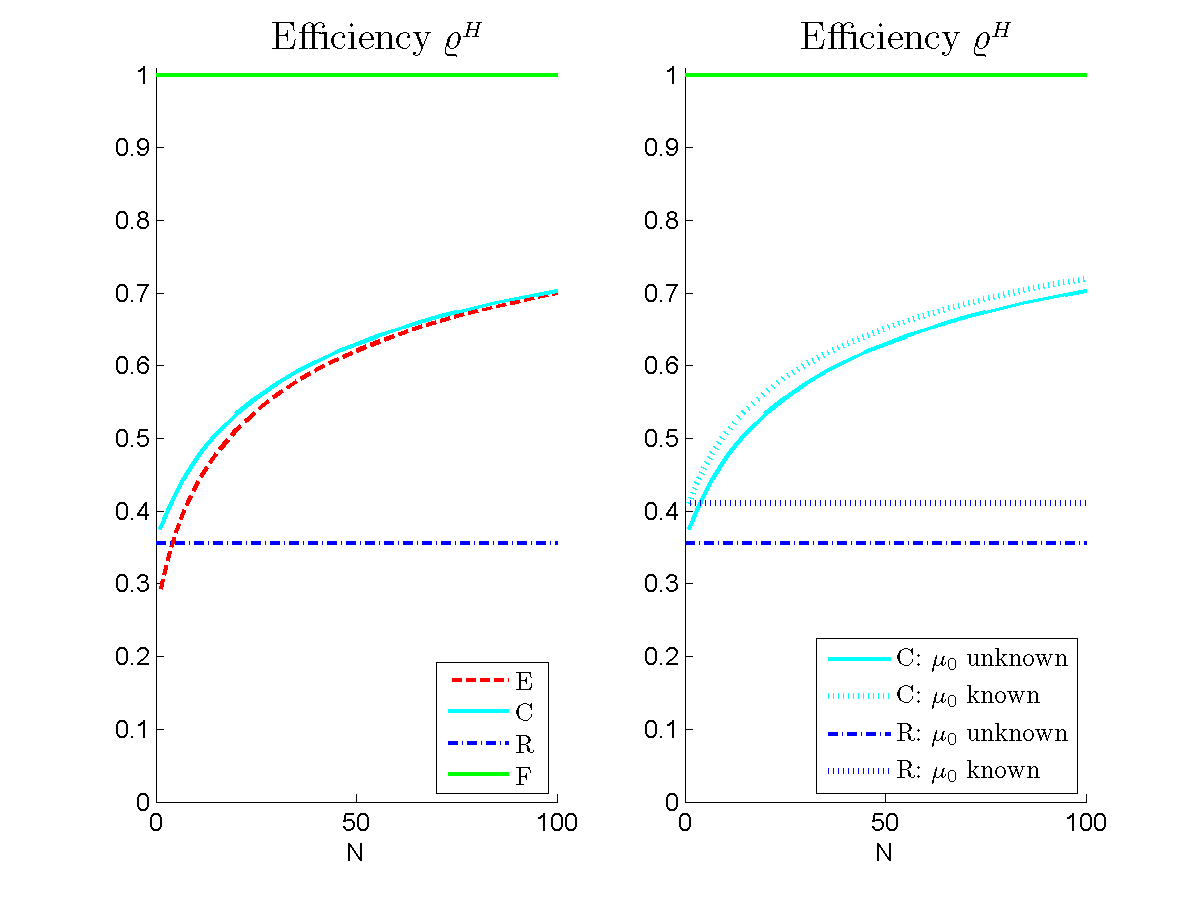}%
\vspace*{-2ex}
 \else \hspace*{-2mm}
\includegraphics[width=120mm,height=70mm]{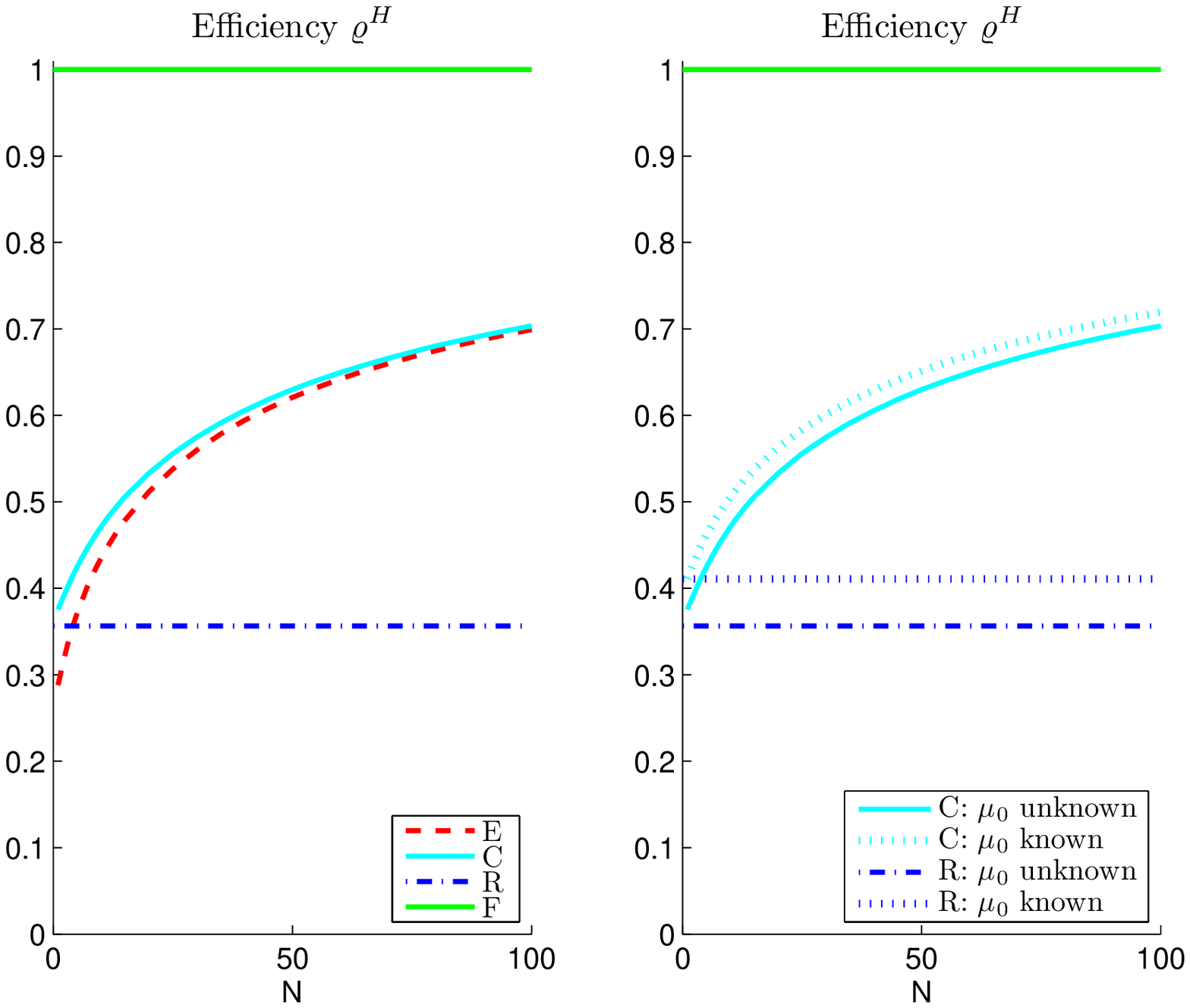}%
\vspace*{-2ex}
\fi

 \centering \caption{\label{efficieny1}
 \small
 Efficiency $\varrho^H$ as a function of the number  $N$
 \newline
  left: \  unknown initial value   ${\mu}_0\sim\mathcal{N}(\delta, \frac{\beta^2}{2\alpha})$ (stationary distribution)\newline
right:  unknown \hspace*{9mm}vs.~known initial value, i.e. \newline
\hspace*{10mm} ${\mu}_0\sim\mathcal{N}(\delta, \frac{\beta^2}{2\alpha})$ vs. ${\mu}_0=\delta=0.05$
 }
\end{figure}

Figure \ref{efficieny1} shows the efficiency $\varrho^H$ as a function of the
number $N$ of information dates.  The left panel shows the results for the different investors assuming the initial value $\mu_0$ is unknown and its distribution is the stationary distribution.
 If an investor starts instead with a known initial value for the drift, i.e.~$\nu_0=\gamma_0^H=0$, then this additional information  increases the value $V^H$ and as a consequence the efficiency $\varrho^H$. This effect is shown for $H=C,R$ in the right panel. For the sake of better comparison we have set the known initial value $m_0$ equal to the mean of the drift $\delta$. So at time $t=0$ the filters are initialized with $\widehat{\mu}^H_{0-}=\delta$, $H=E,R,C$, for known as well as unknown $\mu_0$.

Obviously, we have $\varrho^F=1$ and  the efficiency $\varrho^R$ of the investor observing only returns and no expert opinions is not affected by $N$ while $\varrho^E$ and $\varrho^C$  increase with
$N$. As a consequence of Corollary \ref{prop_value} we always have $\varrho^E\le \varrho^C$ and $\varrho^R\le \varrho^C$, since the investment decisions of the $C$-investor are based on the observation of  returns as well as expert opinions while the $E$-- and $R$--investor have access to only one of these sources of information.

In order to illustrate the asymptotic results for $N\to\infty$ given in Corollary \ref{prop_value_asymN} Table \ref{asymptotics_N} gives for increasing numbers $N$ of information dates the values  $V^{H,N}(1)$ and efficiencies $\varrho^{H,N}$ for $H=E,C$ and compares with the values $V^{H}(1)$ and $\varrho^{H}$ for $H=R,F$. It can be observed  that the result for the fully informed investor ($H=F$) is obtained
for $N\to \infty$, i.e.~$\varrho^{E,N}\to 1$ and $\varrho^{C,N}\to 1$.
For this study we assume that the initial value $\mu_0$ is unknown and its distribution is the stationary distribution. Then according to Theorem  \ref{opt_value_theorem} the value is  $V^{H}(1)= \frac{1}{2\sigma^2}((\delta^2 + {\beta^2}/(2\alpha))T - B^H)$. For the interpretation of the values of $N$ given in Table \ref{asymptotics_N} we note, that for $T=1$ year expert opinions arriving every month, week, day, hour, minute or second  corresponds to $N=12, 52, 365, 8.760, 525.600$ or $31.536.000$, respectively.

\begin{table}
\[
\begin{array}{|r|r|r||r|r|}
\hline
\rule{0mm}{2.5ex}
& \multicolumn{2}{c||}{V^H(1)}&
  \multicolumn{2}{c|}{\varrho^H} \\ \hline
\rule{0mm}{2.5ex}
R &\multicolumn{2}{c||}{0.3213}&
   \multicolumn{2}{c|}{35.63}\\ \hline
\rule{0mm}{2.5ex}
 N           & \hspace*{10mm} E~~~~\  & \hspace*{10mm} C~~~~\  & \hspace*{7mm} E \hspace*{2mm} & \hspace*{7mm} C \hspace*{2mm} \\\hline
 10          & 0.5208  &  0.6008 & 43.49 & 47.12\\
 100         & 0.9957  &  1.0017 &69.94  &  70.36\\
 1.000       & 1.2297  &  1.2299 &88.37  & 88.39\\
 10.000      & 1.3134  &  1.3134 &96.09  & 96.09\\
 100.000     & 1.3407  &  1.3407 &98.74  & 98.74\\
 1.000.000   & 1.3493  &  1.3493 &99.60  & 99.60\\
 10.000.000  & 1.3521  &  1.3521 &99.87  &  99.87\\
   \hline
\rule{0mm}{2.5ex}
 F  &\multicolumn{2}{c||}{1.3533}&\multicolumn{2}{c|}{100.00}\\
\hline
\end{array}
\]

\vspace*{1ex}
 \centering \caption{\label{asymptotics_N}
 \small Value $V^H(1)$ and efficiency $\varrho^H$ in \% for various numbers~ $N$ 
 }
\end{table}
Finally, we study the dependence of the efficiency on the reliability of the expert opinions which is measured by the  standard deviation $\sqrt{\Gamma}$ of the views.
  The left panel of Figure \ref{efficieny2} shows the results for an unknown  initial value $\mu_0$ while the right panel compares  the efficiencies  $\varrho^C$ and $\varrho^R$ for known and unknown initial value.
As in Figure \ref{efficieny1}, the efficiency $\varrho^R$  is not affected by the expert opinions and does not depend on $\sqrt{\Gamma}$  while $\varrho^E$ and $\varrho^C$  decrease with
$\sqrt{\Gamma}$, i.e.~with decreasing reliability of the expert opinions. As before we have $\varrho^E\le \varrho^C$.  The figure also
indicates that for $\Gamma\to \infty $ the efficiency $\varrho^C$ tends to the efficiency of the  $R$-investor since the expert opinions carry no additional information about the drift.

\begin{figure}[h]
\ifpdf \hspace*{-10mm}
\includegraphics[width=140mm,height=70mm]{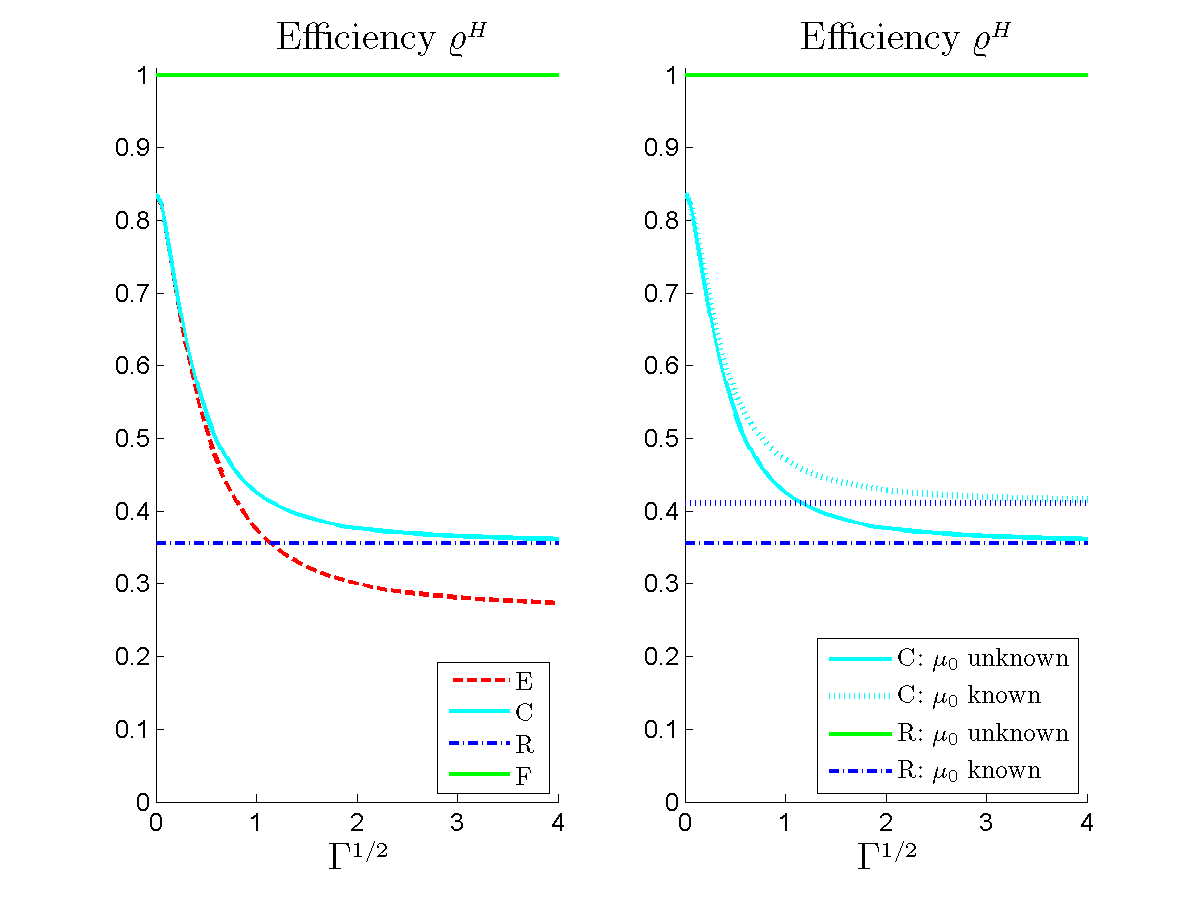}%
\vspace*{-2ex}
 \else \hspace*{-2mm}
\includegraphics[width=120mm,height=70mm]{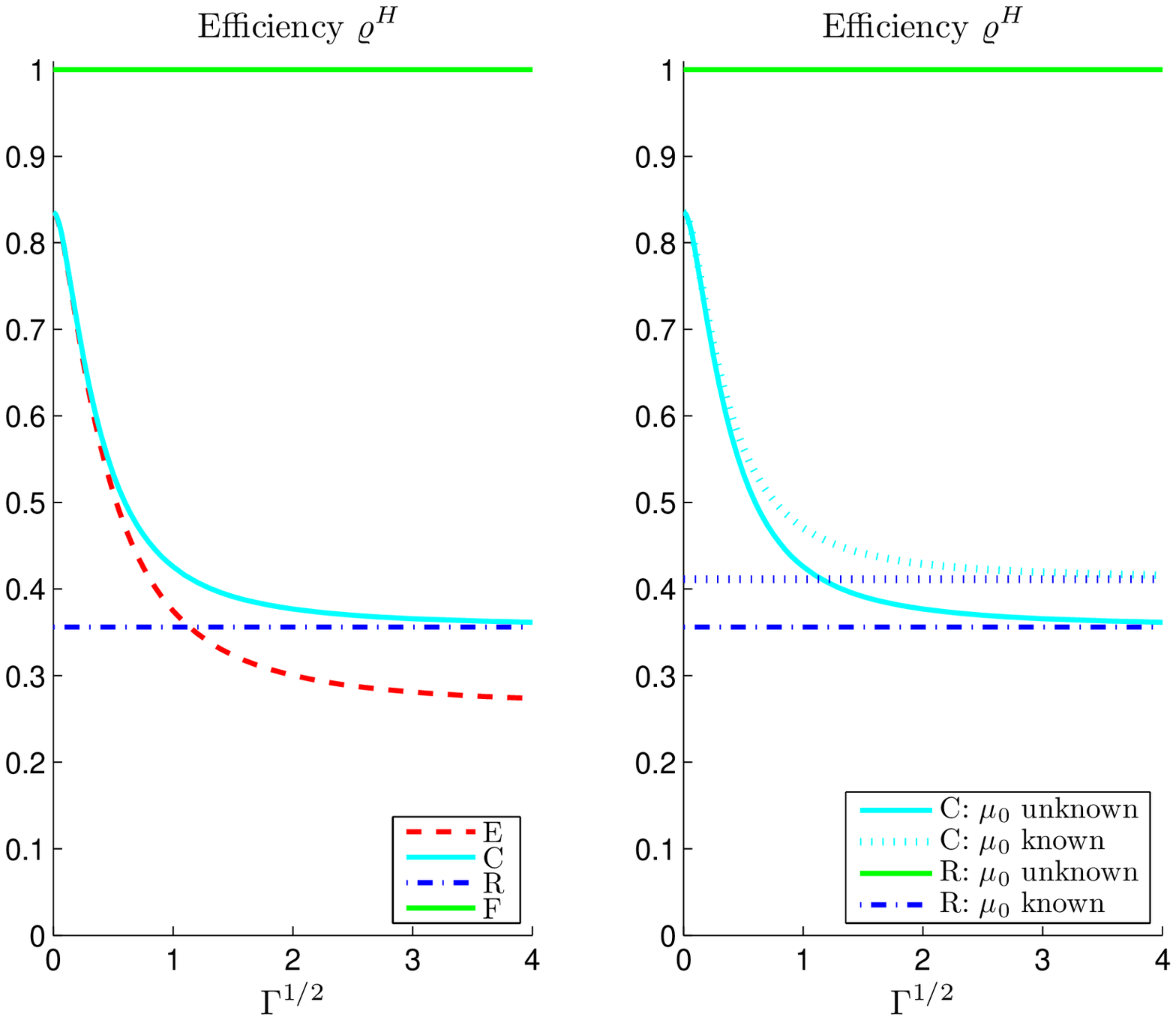}%
\fi

 \centering \caption{\label{efficieny2}
 \small
 Efficiency $\varrho^H$ as a function of $\sqrt{\Gamma}$ (experts reliability) 
  left: \  unknown initial value   ${\mu}_0\sim\mathcal{N}(\delta, \frac{\beta^2}{2\alpha})$ (stationary distribution)\newline
right:  unknown \hspace*{9mm}vs.~known initial value, i.e. \newline
\hspace*{10mm} ${\mu}_0\sim\mathcal{N}(\delta, \frac{\beta^2}{2\alpha})$ vs. ${\mu}_0=\delta=0.05$
\newline
Parameters as in Table \ref{parameter}, $N=20$
 }
\end{figure}

\begin{small}
\bibliographystyle{amsplain}

\end{small}

\end{document}